\documentclass[lettersize,journal]{IEEEtran}
\usepackage{amsmath,amsfonts}
\usepackage{array}
\usepackage[caption=false,font=normalsize,labelfont=sf,textfont=sf]{subfig}
\usepackage{stfloats}
\usepackage{url}
\usepackage{verbatim}
\usepackage{graphicx}
\usepackage{cite}
\usepackage{cite}
\usepackage{amsmath,amssymb,amsfonts,amsthm}
\newtheorem{theorem}{Theorem}

\usepackage{graphicx}
\usepackage{textcomp}

\def\BibTeX{{\rm B\kern-.05em{\sc i\kern-.025em b}\kern-.08em
    T\kern-.1667em\lower.7ex\hbox{E}\kern-.125emX}}

\usepackage{xspace}
\usepackage{siunitx}
\usepackage{caption} 

\usepackage[
]{subcaption}

\usepackage{colortbl}
\usepackage{tabularx}

\usepackage[textsize=tiny]{todonotes}
\usepackage{enumitem}

\usepackage{booktabs}
\usepackage{makecell}
\usepackage{subfig}

\usepackage{flushend}

\usepackage{mathtools}  
\usepackage{amsmath}
\usepackage{amssymb}
\usepackage{tabulary}
\usepackage{booktabs}

\usepackage[unicode]{hyperref}
\hypersetup{pdfstartview=FitH,
  pdfpagelayout=SinglePage,
  bookmarksnumbered=true,
  bookmarksopen=true,
  colorlinks=true,
  linkcolor=violet,
  citecolor=violet
}

\usepackage{algorithm}
\usepackage{algpseudocode}
\newcommand{\MyComment}[1]{\State{\textcolor{gray}{// #1}}}

\newcommand{\ours}{Swift\xspace}
\newcommand{\sync}{Swift-sync\xspace}
\newcommand{\route}{Swift-route\xspace}
\newcommand{\forward}{Swift-forward\xspace}

\begin{document}

\title{Data On the Go: Seamless Data Routing for Intermittently-Powered Battery-Free Sensing}

\author{Gaosheng Liu,~\IEEEmembership{Member, IEEE,} 
        Lin Wang,~\IEEEmembership{Senior Member, IEEE}
\thanks{This work has been partially funded by the Dutch Research Council (NWO) grant OCENW.XS22.1.135. Gaosheng Liu is funded by the China Scholarship Council (CSC) fellowship. (Corresponding author: Lin Wang)}
\thanks{Gaosheng Liu is with the Department of Computer Science, Vrije Universiteit Amsterdam (email: g.s.liu@vu.nl).}
\thanks{Lin Wang is with the Department of Computer Science, Paderborn University and the Department of Computer Science, Vrije Universiteit Amsterdam (email: lin.wang@uni-paderborn.de).}
}

\markboth{IEEE Transactions on Mobile Computing,~Vol.~XX, No.~X, XXXX~XXXX}
{Gaosheng Liu and Lin Wang: Data On the Go}

\maketitle

\begin{abstract}
The rising demand for sustainable IoT has promoted the adoption of battery-free devices intermittently powered by ambient energy for sensing. However, the intermittency poses significant challenges in sensing data collection. Despite recent efforts to enable one-to-one communication, routing data across multiple intermittently-powered battery-free devices, a crucial requirement for a sensing system, remains a formidable challenge.

This paper fills this gap by introducing \ours, which enables seamless data routing in intermittently-powered battery-free sensing systems. \ours overcomes the challenges posed by device intermittency and heterogeneous energy conditions through three major innovative designs. First, \ours incorporates a reliable node synchronization protocol backed by number theory, ensuring successful synchronization regardless of energy conditions. Second, \ours adopts a low-latency message forwarding protocol, allowing continuous message forwarding without repeated synchronization. Finally, \ours features a simple yet effective mechanism for routing path construction, enabling nodes to obtain the optimal path to the sink node with minimum hops. We implement \ours and perform large-scale experiments representing diverse real-world scenarios. The results demonstrate that \ours achieves an order of magnitude reduction in end-to-end message delivery time compared with the state-of-the-art approaches for intermittently-powered battery-free sensing systems.
\end{abstract}


\section{Introduction}
\label{sec:introduction}

\IEEEPARstart{B}{attery}-free (BF) devices have been increasingly adopted for sensing due to their maintenance-free and environment-friendly traits~\cite{2017-sensys-futureic,2017-snapl-survey,2022-tecs-camaroptera, 2018-mobicom-eyetracker,2021-weee-cic,2020-ipsn-cis,2020-sensys-bfsensing,2023-ipsn-riotee}. These devices harvest energy from ambient sources like light, vibrations, and radio waves~\cite{2017-imwut-bfphone,2020-imwut-bfgame,2018-asplos-energystore}, and buffer it in capacitors to function. However, the ambient energy is usually limited and cannot sustain continuous device operation. Consequently, BF devices alternate between short periods of activity (e.g., a few milliseconds) and much longer periods of inactivity (e.g., 100s of milliseconds or seconds) for recharging, following a charging cycle determined by ambient energy conditions. To overcome this limitation, recent efforts have focused on achieving continuous computing on a single BF device~\cite{2016-osdi-ratchet,2016-oopsla-chain, 2018-sensys-ink,2020-sensys-state,2020-pacmpl-foundation,2020-pldi-scheduling,2020-tecs-checkpoint,2020-asplos-tics,2020-imwut-bfree,2022-pldi-wario} and enabling one-to-one communication~\cite{2022-mobisys-bluetooth,2021-nsdi-find,2022-nsdi-bonito}. Despite these advances, a significant challenge remains unsolved: the issue of \emph{data routing}, which involves transferring data across a network of intermittently-powered BF devices, a crucial function for sensing data collection.

Efficient data routing faces significant challenges imposed by the \emph{intermittency} of BF devices. In stark contrast to conventional duty cycling, where device states (e.g., active and sleep) are intentionally controlled~\cite{2019-sigcomm-nd,2016-infocom-panda,2016-rtss-eh}, intermittency introduces \emph{unpredictable} impacts on device operation. This uncertainty makes even simple tasks like maintaining reachability, which are straightforward for continuously-powered devices, much more complex. Synchronizing the working periods of two BF devices to ensure they are powered on simultaneously becomes a daunting task without knowledge of each other's charging cycle. Since BF devices spend most of their charging cycles turned off, the likelihood of both devices being powered on simultaneously is slim. Moreover, the heterogeneity of energy availabilities in different locations leads to diverse charging cycles for BF devices, requiring repeated synchronization attempts to maintain continuous reachability. Unfortunately, such synchronization efforts incur significant time overhead, proving detrimental for data routing, which relies on seamless, continuous data flow through the network, hop-by-hop.

Despite the extensive literature on routing~\cite{2012-csur-wsn}, current solutions suffer from inefficiencies due to their inability to handle device intermittency, resulting in significant message delivery delays caused by frequent node synchronization. Moreover, BF devices are energy-constrained and lack powerful geo-location components commonly found in continuously-powered devices, which are instrumental in efficient routing~\cite{2012-csur-wsn}. While energy-harvesting wireless sensor networks are closely related~\cite{2018-tsn-ehwsn,2021-tgcn-renew}, their continuous operation based on active power management and duty cycling~\cite{2014-mobihoc-eh,2016-infocom-panda,2018-infocom-ehwsn,2019-sigcomm-nd} sets them apart from sensing systems built with BF devices that operate intermittently on a millisecond timescale. These factors make existing routing approaches either inapplicable or inefficient, emphasizing the need for novel solutions.

In this paper, we propose \ours
, an efficient routing scheme that enables seamless data routing for sensing systems built with intermittently-powered BF devices. \ours incorporates three major innovations to address the challenges related to node synchronization, message forwarding, and route construction. First, \ours tackles the challenge of synchronizing BF devices under diverse energy conditions with \sync, a reliable synchronization protocol inspired by number theory principles. Specifically, we transform the node synchronization problem into an instance of the linear congruential generator (LCG) problem. \sync adjusts the working period of the sending BF device following the LCG and guarantees reliability by ensuring the full-cycle property of the LCG. With \sync, two BF devices without shared knowledge can reliably discover each other within a fixed time frame.

Second, \ours incorporates a low-latency message forwarding mechanism called \forward, specifically designed to facilitate efficient message passing for BF devices along a given path. The main challenge lies in the repeated synchronization between neighboring nodes at each hop, which leads to prolonged message delivery latency. We observe that once two BF devices are synchronized with \sync, they can maintain their synchronized state without the need for further expensive synchronizations. Building upon this finding, \forward optimizes message forwarding by caching the charging cycle of the next hop at each node upon successful synchronization and tracking the offset between the node's working time and that of its next hop. With \forward, each node can synchronize with its next hop and return to its original state, all within a single charging cycle time. The synchronization overhead is significantly reduced, thus greatly improving the efficiency of message forwarding.

Finally, \ours accomplishes route construction with a protocol called \route, which relies on a flooding mechanism initiated by the sink node at its core. BF devices announce their distance (i.e., hop count) to the sink node to their neighbors and update their next hop if they receive an announcement with a smaller hop count. To achieve an organized hop count propagation, \route employs a carefully crafted timing strategy that creates a layered structure based on each node's proximity to the sink node. With that, \route ensures the optimal construction of routes for each node towards the sink in terms of hop count and is easy to implement.

Overall, we make the following contributions in this paper. After identifying the key challenges for seamless routing in BF sensing systems (\S\ref{sec:motivation}), we present \ours where we 
\begin{itemize}
    \item propose a reliable node synchronization protocol \sync, guaranteeing successes within a fixed amount of time, regardless of the energy conditions (\S\ref{sec:sync}).
    \item introduce a low-latency message forwarding mechanism tailored for intermittently-powered BF devices, significantly reducing the synchronization overhead (\S\ref{sec:forward}). 
    \item present a route construction protocol \route for BF sensing systems, allowing each node to obtain a path to the sink with the least hop count (\S\ref{sec:route}). 
    \item implement \ours in OMNeT++ and perform large-scale experiments to validate its performance under various real-world scenarios (\S\ref{sec:eval}). We have also tested part of \ours on a hardware testbed. Results show that \ours achieves an order of magnitude gains over the state-of-the-art approaches in reducing the message delivery time.
\end{itemize}
\S\ref{sec:relatedwork} discusses related works and \S\ref{sec:conclusions} concludes the paper.

\section{Background and Motivation}
\label{sec:motivation}
In this section, we introduce the background on BF sensing, identify the challenges brought by BF devices in sensing systems, and motivate our work.

\subsection{Intermittently-Powered BF Sensing}
\label{sec:motivation:sensing}
Advancements in energy-harvesting technologies have significantly spurred the adoption of BF devices for the development of sensing systems, particularly in challenging environments~\cite{2017-snapl-survey,2019-asplos-genesis,2021-weee-cic,2020-ipsn-cis,2020-sensys-bfsensing}. These BF devices are equipped with energy harvesters capable of capturing ambient energy such as energy from the sun, vibrations, wind, and radio waves, which is then buffered in capacitors to power the device~\cite{2018-asplos-energystore}. Since they operate passively without batteries, such sensing systems, once deployed, can last long without maintenance. Also, battery-freeness makes them more environment-friendly~\cite{2017-snapl-survey}.

BF devices work \emph{intermittently}---the device suffers power failures and alternates between on and off states constantly. This is because the harvested energy buffered in the tiny capacitor can only support the running of the device---a low-power microcontroller (e.g., TI MSP430-series)---for a very short period (e.g., a few milliseconds or less)~\cite{2022-tecs-energyconsumption}. On the other hand, charging the capacitor takes significantly longer (e.g., 100s of milliseconds or even seconds) due to the scarcity of ambient energy~\cite{2022-nsdi-bonito}. A BF device is powered on when the capacitor is charged to a certain level and suffers a power failure when the charging level of the capacitor drops below a threshold. We call the time when the device is on the \emph{working period} and the \emph{charging period} otherwise. A \emph{charging cycle} spans one charging period and one working period. 

BF sensing systems exhibit distinctive characteristics that differentiate them from traditional wireless sensor networks (WSNs) and energy-harvesting WSNs (EH-WSNs). The main differences between these systems are summarized in Table~\ref{tab:bf-diff}. WSNs and EH-WSNs, typically powered by relatively rich energy supplies (e.g., rechargeable batteries and/or large capacitors), adopt active power management and duty cycling to control the device state (e.g., active, sleep) intentionally~\cite{2014-mobihoc-eh,2016-infocom-panda,2018-infocom-ehwsn,2019-sigcomm-nd}. On the other hand, BF sensing systems rely on devices powered by the scarce energy buffered in small capacitors (at the $\mu$F level) passively, resulting in an \emph{intermittent} operational style at the millisecond scale induced by frequent power failures. This intermittency poses challenges for programs running on BF devices, as progress can be hindered if intermediate execution states are not preserved correctly. To address this issue, researchers have proposed two main categories of techniques: checkpointing program state~\cite{2020-asplos-tics, 2019-pldi-samoyed,2022-pldi-wario} and transforming programs into idempotent tasks~\cite{2020-asplos-chrt,2016-oopsla-chain,2018-sensys-ink,sensys-2017-task}. Although these efforts focus on individual BF devices, they have established a solid foundation for the adoption of intermittently-powered BF devices in sensing applications.

\begin{table*}[!t]
\small
    \caption{BF Sensing Systems vs. Traditional WSNs}\label{tab:bf-diff}
    \centering
    \begin{tabular}{llll}
        \toprule
         & \textbf{WSNs} & \textbf{Energy-Harvesting WSNs} & \textbf{Battery-Free Sensing} \\
        \midrule
        \rowcolor[gray]{0.95}\textbf{Energy storage} & Batteries & Rechargeable batteries and/or large capacitors (mF/F) & Small capacitors ($\mu$F) \\
        \textbf{Power management} & Active & Active & Passive \\
        \rowcolor[gray]{0.95}\textbf{Operational style} & Duty-cycling, predicable & Duty-cycling, predictable & Intermittent, unpredictable \\
        \textbf{Working time scale} & Years & From minutes to hours & Milliseconds \\
        \bottomrule
    \end{tabular}
\end{table*}

\subsection{Synchronizing Intermittently-Powered BF Devices}
\label{sec:motivation:sync}
To establish communication between two intermittently-powered BF devices, they need to be in the active state simultaneously, necessitating the synchronization of their working periods. The challenge lies in the intermittent nature of their activity and the limited shared knowledge between the two devices, further complicating the synchronization process.

In scenarios where the energy condition is homogeneous across all devices, i.e., all BF devices operate on the same charging cycle, the problem becomes determining the offset between the working periods of two devices. To this end, Find utilizes a randomized approach to identify the working period offset between two devices~\cite{2021-nsdi-find}. In each charging cycle, the device employs a randomized approach to introduce a delay to its working period. This delay is generated from a well-tuned geometric distribution, allowing synchronization in a reasonable number of charging cycles on average. However, due to the random nature of this process, certain cases might require a significant or even infinite number of charging cycles for successful synchronization. In practice, BF devices are distributed geographically over a wide area. The heterogeneous energy conditions in different locations make it improbable for all BF devices to adhere to the same charging cycle. Consequently, ensuring constant reachability after synchronization becomes uncertain. To address this challenge, Bonito adopts a learning-based approach to predict the charging cycles based on statistical models~\cite{2022-nsdi-bonito}. In a nutshell, these approaches do not ensure successful synchronization, hence providing no reliability guarantee.

\subsection{Data Routing for Intermittently-Powered BF Sensing}
\label{sec:motivation:routing}
Typically, sensing systems collect data from various nodes and send it to a central sink node for processing and analysis. However, due to limited communication coverage, not all nodes can directly communicate with the sink. Consequently, messages generated by distant nodes need to be relayed through intermediate nodes to reach the sink. This challenge is the essence of the routing problem~\cite{2012-csur-wsn}. The main objective in routing is to find an optimal path that efficiently delivers messages to the sink in a timely manner.

The current routing protocols are not well-suited for intermittently-powered BF sensing due to the following reasons: (1) Existing protocols designed for WSNs or EH-WSNs rely on active duty cycling to regulate device state intentionally~\cite{2012-csur-wsn, 2018-tsn-ehwsn}. However, this approach is unsuitable for intermittently-powered BF devices because they have limited power supply and unpredictable power outages. (2) Most routing protocols do not consider the time overhead involved in synchronizing two intermittently-powered BF devices, which results in high message delivery latency. (3) Protocols developed for continuously-powered devices often involve complex logic and utilize expensive geo-location peripherals, making them unaffordable for intermittently-powered BF devices. 

Efficient data routing for intermittently-powered BF sensing involves addressing the following challenges:
\begin{enumerate}
    \item \emph{Minimizing node synchronization overhead:} To achieve efficient routing, messages need to be forwarded hop-by-hop on the path with minimal node synchronization overhead. Constantly synchronizing with the next hop for every transmission would be prohibitively costly for intermittently-powered BF devices.
    \item \emph{Optimal path selection:} Each node must be able to determine an efficient path towards the sink without knowing the location of other nodes. The goal is to find a path with the least number of hops required for the node to reach the sink. This way, data can be efficiently relayed to the central sink node.
\end{enumerate}

In light of the above challenges, we aim to design a routing scheme specific to intermittently-powered BF sensing, featuring mechanisms for node synchronization, message forwarding, and route construction.

\section{Reliable Node Synchronization}
\label{sec:sync}
In this section, we clarify the system model and introduce our synchronization algorithm based on the number theory.

\subsection{System Model}
\label{sec:sync:model}
We consider an intermittently-powered BF system comprising a fixed number of BF devices distributed across a large geographical area. These devices (nodes) possess identical hardware, including components like microcontrollers, capacitors, and energy-harvesting units, thus having roughly the same working period duration. We divide time into time slots where the working period for each node spans one time slot, while the charging period extends over multiple time slots since it is substantially longer than the working period as mentioned. The duration of the charging period, denoted as $t$, may differ among nodes based on their respective energy conditions. It might also fluctuate over time but usually at a much larger time scale (e.g., minutes). In practice, $t$ usually falls within the range of $[5, 500]$ ms with capacitors of 10s of $\mu$F typically used on BF devices and the difference in charging time between any two nodes can be bounded by a constant in a given environment, denoted as $\alpha$. In a typical stable environment without abnormal disruptions, $\alpha$ can be set to a small number~\cite{2021-nsdi-find, 2022-nsdi-bonito}, and in our setting, we use $\alpha=3$. This means the maximum charging time of a BF device is at most three times that of any other devices in the sensing system. Note that our proposed designs also work under more hostile environments with larger values for $\alpha$, albeit with possibly reduced efficiency.

In this work, we assume a time slot lasts for one millisecond, as also seen in existing works~\cite{2021-nsdi-find, 2022-nsdi-bonito}. We argue that this period is long enough for successful communication (e.g., sending a message out and receiving an acknowledgment back) between a pair of battery-free devices if their working periods are perfectly aligned. For example, BLE 5.0 is capable of transmitting up to 175 bytes within one millisecond~\cite{2021-specification-BLE5}. The slot length may deviate from $1$ ms in real-world scenarios. As long as a round-trip communication between two BF devices can be completed within one slot, our mode will apply.
However, achieving perfect alignment of time slots among battery-free devices is a challenging task in itself. Existing works have attempted to use an external light source that is perceivable at all battery-free devices for time slot alignment~\cite{2021-nsdi-find}, which may not always be available in different sensing scenarios. We assume such a time slot alignment mechanism is used whenever it is feasible. Nevertheless, our model and solutions can be generalized to cases without external mechanisms for time slot alignment. Assume the communication between two BF devices can be successful when the overlap of the working periods of these devices exceeds 50\% of one time slot (i.e., 500 $\mu$s), which is quite practical with BLE 5.0. By trying one more charging cycle, it is always ensured that two BF devices can communicate if we delay the working period of one of the BF devices by one time slot in the second charging cycle. Once one communication round is performed, the two BF devices can infer their time slot offset by monitoring the difference between their working time and communication time. For simplicity, we will use aligned time slots to describe our protocol design.

\begin{figure*}[t]
    \centering
    \includegraphics[page=1,width=0.85\textwidth]{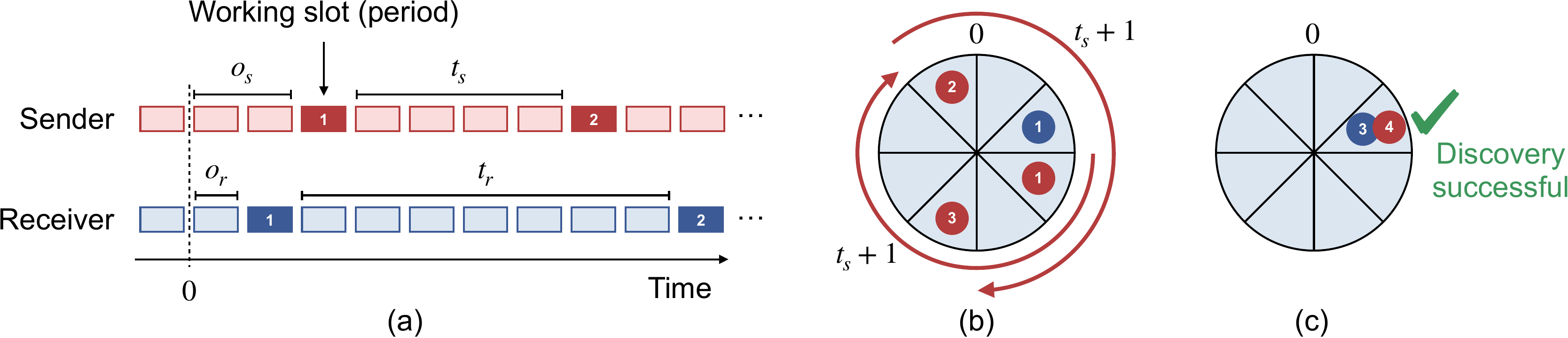}
    \caption{Illustration of the node synchronization procedure: (a) charging cycles of the sender and the receiver over time (boxes represent time slots), (b) the offsets of the sender's working periods in relation to the receiver's charging cycle hypothetically, and (c) successful discovery when the working period offsets of the sender and receiver match.}
    \label{fig:sync:procedure}
\end{figure*}

\subsection{Theoretical Analysis}
The Find approach shows statistically good node synchronization performance for cases with homogeneous energy conditions~\cite{2021-nsdi-find}. Extending it to heterogeneous cases is possible, but it has two main limitations: First, fine-tuning the randomization parameters for optimal performance becomes challenging, as parameters set for one neighbor may negatively affect others when a node synchronizes with multiple neighbors. Second, Find does not offer a guaranteed synchronization outcome due to its random nature. To address these limitations, we propose a deterministic approach that ensures a guaranteed success rate for node synchronization. Before presenting our design, we provide a theoretical analysis of the node synchronization problem.

\begin{table}[!t]
\caption{Notation} \label{tab:notation}
\centering 
\begin{tabular}{lp{6cm}}
\toprule
\textbf{Symbol} & \textbf{Meaning} \\
\midrule
$o_s$ / $o_r$ & Initial charging time offset for the sender/receiver\\
$t_s$ / $t_r$ & Charging time of the sender/receiver \\
$c\_inc$ & Current increment for co-prime\\
$\alpha$& Maximum difference factor for the charging time\\
$\delta$& Maximum increment for co-prime lookup\\
$b_{s}$ / $b_{r}$ & Bias of the sender working period with respect to the original working period of the sender/receiver\\
$d_{s}$ & Delay to apply to the sender to remain synchronized\\
$t\_dis$ & Discovery time for a BF device to find a neighbor\\
$t\_self$ & Charging time of a BF device\\
$t\_max$ & Maximum charging time of a BF device\\
$t\_wait$&  Remaining listening (waiting) slot for the neighbors \\
\bottomrule
\end{tabular}
\label{tab:TableOfNotationForMyResearch}
\end{table}

For a pair of nodes trying to synchronize, we call the node sending out messages the \emph{sender} and the one receiving messages the \emph{receiver}. Figure~\ref{fig:sync:procedure} depicts the charging cycles of the sender and receiver, along with the synchronization procedure. We select a random time reference (time $0$ shown in the figure) and denote the distance of the working period of a node from this reference as the initial offset. This initial offset is denoted by $o_s$ and $o_r$ for the sender and receiver, respectively. Figure~\ref{fig:sync:procedure}(a) shows that following the initial offset, the sender follows a charging cycle of length $t_s+1$ while the receiver follows a charging cycle of $t_r+1$. 
To facilitate understanding, we list all the notation throughout the paper in Table~\ref{tab:notation}.

One convenient way to represent the synchronization procedure is to choose the charging cycle of the receiver as a reference \emph{hypothetically} and analyze the offsets of the working periods of the sender on the receiver's charging cycle, as shown in Figure~\ref{fig:sync:procedure}(b). For the receiver, the working period offsets can be calculated as $(o_r + n \times (t_r + 1)) \% (t_r + 1)$, with $n \in \mathbb{N}$, which is always equal to $o_r$. That is, the location of the working periods of the receiver in Figure~\ref{fig:sync:procedure} is fixed. For the sender, the offsets of its working periods with respect to the charging cycle of the receiver can be calculated as  $(o_s + n \times (t_s + 1)) \% (t_r + 1)$ which also provides the positions of the working periods of the sender in Figure~\ref{fig:sync:procedure}. Intuitively, the sender cycles through the charging cycle of the receiver with a step of size $(t_s + 1) \% (t_r + 1)$. 

The goal is to find a time slot where both the sender and the receiver are in the working state. For the example shown in Figure~\ref{fig:sync:procedure}, the sender and receiver can communicate with each other in the fourth working period of the sender and the third working period of the receiver, as shown in Figure~\ref{fig:sync:procedure}(c).

An important question to consider is whether synchronization is guaranteed to succeed within a fixed amount of time. The answer, however, is affirmative only in specific cases. In particular, we demonstrate that:

\begin{theorem}
\label{thm:motivation:coprime}
    The node synchronization is guaranteed to be successful if $t_s+1$ and $t_r+1$ are coprime. 
\end{theorem}
\begin{proof}
    We prove this by mapping the node synchronization problem to the linear congruential generator (LCG) problem. We refer readers to Chapter 3 of~\cite{1997-art-lcg} for a detailed explanation of LCG. An LCG is an algorithm that produces a sequence of pseudo-randomized numbers defined by the recurrence relation:
    \begin{equation}
        X_{n+1} = (aX_n + c) \mod m,
    \end{equation}
    where $X$ is the sequence of the pseudo-random numbers, $m$ is the modulus, and $a \in (0,m)$ and $c \in [0,m)$ the multiplier and the increment, respectively. $X_0 \in [0,m)$ is the seed for the sequence generation. For an LCG to obtain its maximum length $m$, a.k.a. full cycle LCG, it must satisfy the following conditions according to the Hull-Dobell Theorem~\cite{1962-siam-generator}:
    \begin{enumerate}[label={\textbf{C\arabic*:}}]
        \item $c$ and $m$ are coprime;
        \item $a-1$ is divisible by all prime factors of $m$; and
        \item if $4$ divides $m$, then $4$ divides $a-1$. 
    \end{enumerate}
    The node synchronization problem can be reduced to the LCG problem if we set $a=1$, $c = t_s+1$, and $m=t_r + 1$. $X$ represents the offset of the working period of the sender in relation to the charging cycle (i.e., $t_r+1$) of the receiver. In this case, obtaining the maximum length for this LCG means that $X$ will traverse all possible numbers in the range $[0, t_r+1)$. This ensures that starting from seed $X_0 = o_s$, for some $n \in \mathbb{N}$ we will have $X_n = o_r$. In the context of the synchronization problem, this means that there will always be a case where the working period offsets of the sender and the receiver match with each other. This requires satisfying the above conditions for the Hull-Dobell Theorem where \textbf{C2} and \textbf{C3} are satisfied trivially and according to \textbf{C1} we require $t_s+1$ and $t_r+1$ to be coprime, hence completing the proof.
\end{proof}

This result indicates that node synchronization is naturally guaranteed only when $t_s+1$ and $t_r+1$ are coprime. However, if this condition is not satisfied, the synchronization guarantee is lost, and more importantly, the nodes are not able to determine whether the condition holds or not in advance. Therefore, to ensure successful synchronization in all cases, a mechanism must be devised. In the following sections, we illustrate how we can achieve this objective by building upon Theorem~\ref{thm:motivation:coprime} and introduce our proposed synchronization protocol.

\subsection{Reliable Node Synchronization with \sync}
We propose \sync, a reliable synchronization protocol with guaranteed success. Our idea is inspired by Theorem~\ref{thm:motivation:coprime}, where in cases of the coprime condition not being met, we come up with a simple trial-and-error strategy to create such a condition artificially. Note that our algorithm does not rely on any assumptions regarding the coprime condition, which will unlikely hold in real-world scenarios. 

The idea of \sync is as follows: We start by assuming that $c = t_s+1$ and $m=t_r+1$ are coprime and try to validate this assumption within a worst-case bound. If the synchronization is not successful within this bound, we increase $c$ by one slot and repeat the above procedure. The worst-case bound is calculated as the maximum number of slots needed for the working period offsets of the sender and receiver to match, assuming $c$ and $m$ are coprime. Following the LCG problem, for every charging cycle of the sender, the working period offset of the sender corresponds to a random number generated by the LCG. In the worst case, $m$ random numbers need to be generated before we match with the working period offset of the receiver. Since the sender has no information about the charging cycle of the receiver in reality, we assume $t_r$ is bounded by constant factor times of $t_s$ where the constant is typically quite small, e.g., three. 
Combined with the charging time range of $[5,500]$ slots, the upper bound for $t_s+1$ is thus $1503$ slots. Through our empirical observation, we notice that the maximum coprime gap for two numbers under $1503$ is $10$.
By increasing $c$ with a step of one slot upon synchronization failure, it is ensured that there will be a case where $c$ and $m$ are coprime within at most ten attempts, meeting the conditions for Theorem~\ref{thm:motivation:coprime} and thus making the synchronization successful. Note that the assumptions of the charging time ranges are realistically derived from real-world conditions~\cite{2021-nsdi-find,2022-nsdi-bonito}.

\begin{algorithm}[!t]
    \small
    \caption{Node synchronization mechanism \ours-sync}\label{algo:sync}
    \begin{algorithmic}[1]
    \State{\textbf{Input:} node's charging time: $t\_{self}$, maximum charging cycle multiplier: $\alpha$, maximum coprime gap: $\delta$}
    \State{\textbf{Output:} discovery time: $t\_{dis}$, neighbor's ID: $id\_{nbr}$, neighbor's charging time: $t\_{nbr}$)}
    
    \State{$c\_inc \gets 0$, $cycle\_count \gets 0$, $flag \gets \textrm{false}$}
    \State{$id\_{nbr} \gets -1$, $t\_{nbr} \gets 0$, $t\_{dis} \gets 0$}
    \While{$flag == \textrm{false} $}
        \MyComment{Maximum tries reached, step forward by one slot}
        \If{$cycle\_count > \alpha (t\_self + 1)$}
            \If{$c\_inc > \delta$}
                \State{No neighbors available, return error}
            \Else
                \State{$c\_inc \gets c\_inc + 1$}
                \State{$cycle\_count \gets 0$}
            \EndIf
        \EndIf
        \State{$cycle\_count \gets cycle\_count + 1$}
        \State{Delay the working period by $c\_inc$ time slots}
        \MyComment{Coprime condition met, working period offsets matched}
        \If{neighbor found}
            \State{$flag \gets \textrm{true}$}
            \State{$t\_dis \gets (c\_inc + t\_self + 1)$}
            \State{Record neighbor's ID $id\_nbr$ and charging time $t\_nbr$}
        \EndIf
    \EndWhile
    \end{algorithmic}
\end{algorithm}

The workings of \sync are listed in Algorithm~\ref{algo:sync}. We keep searching until a neighbor is synchronized (lines 5--23). In every round, we check if we have reached the maximum attempts for the current increment for $c$ (line 7), i.e., $c\_inc$ which is assumed to be bounded by $\alpha$ times the charging cycle of the node itself, $\alpha (t\_{self} + 1)$.
If so, we further check if we have reached the maximum necessary increment for $c$ (line 8), which is bounded by the coprime gap $\delta$. If that is true, we can conclude that there are no neighbors available and the synchronization algorithm terminates with an error (line 9). Within the maximum coprime gap, the increment grows by one (line 11) and we use the new increment to continue the search. In every charging cycle of the node, we increase the number of attempts (line 15) by one and delay the working period of the node by $c\_inc$. If a neighbor node is found, we stop the synchronization process and return the discovery time, the ID of the found neighbor, and the charging time of the found neighbor to facilitate future communications with this neighbor (lines 18--22). If no neighbors have been found after we have incremented $c$ by more than $\delta$, we can conclude that there are no reachable neighbors, and an error is returned. This is because \ours-sync guarantees that if a neighbor is available, it will always be found within the maximum number of attempts---a big advantage of a deterministic algorithm.

\section{Message Forwarding}
\label{sec:forward}
\begin{figure*}[!th]
    \centering
    \includegraphics[page=3,width=0.95\textwidth]{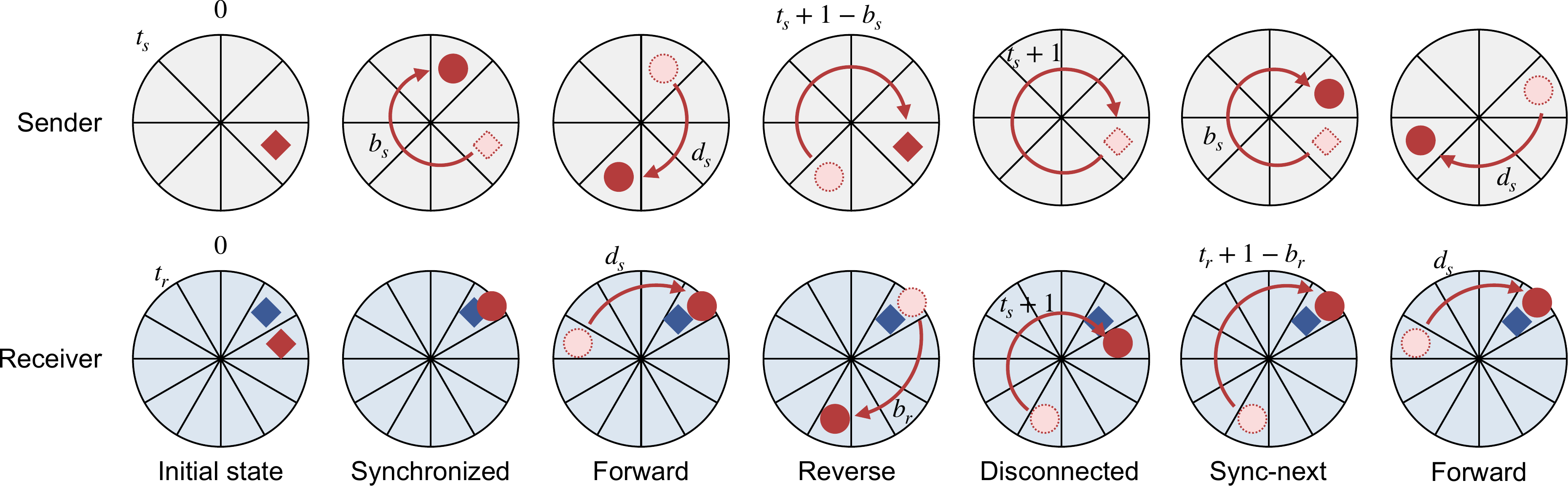}
    \caption{Illustration of the state changes with the \ours-forward mechanism, where we show how the charging cycles of the sender should be adapted to synchronize with its next hop (e.g., the receiver) directly and reverse back when the message forwarding is completed. The receiver is passive and the sender uses the receiver's charging cycle as its reference.}
    \label{fig:forward:states}
\end{figure*}

In a route towards the sink, each node must synchronize with its next hop before sending out messages with existing routing protocols, resulting in significant time overhead. However, we posit that node synchronization is only necessary for the initial communication with the next hop. Once synchronized, the node can retrieve and record the charging cycle information and the relative working period of the next hop. This allows for efficient communication during periods when the charging cycles remain stable, eliminating the need for repeated synchronizations for every hop on the route.

\subsection{High-Level Idea}
\label{sec:forward:idea}
We explain the high-level idea through the example shown in Figure~\ref{fig:forward:states}. We start with the initial state, where both the sender and receiver follow their own charging cycles and work periodically in time slots indicated with the colored diamond (red: sender, blue: receiver). Without loss of generality, we assume $t_s < t_r $ in this example. The zero point is just an arbitrary reference point. We assume the receiver is passive and thus, it does not adjust its working periods in the listening state. We denote by $b_s$ the bias of the working period of the sender w.r.t. its initial working period position. Using the charging cycle of the receiver as a reference, the sender's working slot unlikely matches that of the receiver, and we denote by $b_r$ the bias of the sender's working period w.r.t. the working period of the receiver, reflected on the receiver's charging cycle. Each node in the system can be in one of the following states, with state transitions depicted in Figure~\ref{fig:forward:states}.
\begin{description}
    \item[Synchronized] During synchronization, the sender follows the \ours-sync algorithm (Algorithm~\ref{algo:sync}) where $b_s$ is updated as 
    $b_s = (b_s + t_s + 1 + c\_inc)\%(t_s + 1)$. When the synchronization is completed, the sender's working period is aligned with that of the receiver and we set $b_r = 0$. 
    \item[Forward] Upon synchronization, the sender maintains the synchronized state with the receiver by delaying its working period by $d_s = (t_r - t_s) \% (t_r + 1)$ slot(s) in every charging cycle. As a result, $b_s$ needs to be updated with $b_s = (b_s + d_s) \% (t_s + 1)$ and $b_r$ remains zero. The sender remains synchronized with the receiver in every charging cycle and messages can be delivered between the two nodes.
    \item[Reverse] Once the message forwarding is done, the sender reverses back its working period to its initial state so that it can be synchronized with its upstream nodes. Since $b_s$ keeps a record of the bias of the current working period of the sender w.r.t. the initial value zero, we can simply delay the working period of the sender by $t_s + 1 - b_s$ slot(s) in the next charging cycle. As a result of that, we set $b_r = (t_s + 1 + t_s + 1 - b_s)\%(t_r + 1)$ where the first $(t_s + 1)$ component comes from the fact that one charging cycle is needed for applying the artificial delay and set $b_s = 0$.
    \item[Disconnected] The reverse operation disconnects the two nodes. In this state, the sender may serve as a receiver for an upstream node. To be able to fast sync with its next hop without running \ours-sync again, the sender updates $b_r$ with 
    $b_r = (t_s + 1 + b_r)\%(t_r + 1)$ in each of its charging cycles. 
    \item[Sync-next] When the sender decides to forward messages to its next hop again, the sender simply delays its working period by $t_r + 1 - b_r$ slot(s), and the working periods of the two nodes become aligned. Here, we set $b_s = (t_r + 1 - b_r)\%(t_s + 1)$ and $b_r = 0$. After that, the sender enters the Forward state for delivering messages to the receiver, as already described above.
\end{description}

\subsection{Message Forwarding with the \forward}
\label{sec:forwarding:protocol}
Based on the above concepts, we introduce \forward, a message-forwarding protocol that operates on each node and comprises two key procedures: \emph{sending} and \emph{receiving} during the steady state. In this protocol, nodes alternate their status between these two procedures, as detailed below.

\begin{algorithm}[!t]
\small
\caption{Message forwarding protocol \forward}\label{algo:forward}
\begin{algorithmic}[1]
\State{\textbf{Input:} charging time of the current node and its next hop: $t_s$, $t_r$, current node's working period bias w.r.t. its initial state $b_s$ and its next hop $b_r$, id of the next hop: $id\_next$}
\Procedure{Send}{$queue$}
    \State{$attempt \leftarrow 0$}
    \While{$queue.size \ne 0 $}
        \State{$d_s \gets |t_r - t_s|\%(t_r + 1)$}
        \State{Delay the working period by $d_s$ slot(s) in every cycle}
        \State{$msg \gets queue.\textrm{pop}()$, $msg.id\_next \gets id\_next$}
        \State{Send message $msg$ and wait for ACK}
        \If {ACK timeout (e.g., due to collision)}
            \State{$attempt \leftarrow attempt + 1$, retry}
            \If{$attempt > max\_attempt$}
                \State{Break}
            \EndIf
        \EndIf
        \State{$b_s \gets (b_s + d_s)\%(t_s + 1)$}
        \MyComment{Forwarding done, reverse to receiving}
        \If {$queue.size == 0$}
            \State{Break}
        \EndIf
    \EndWhile
    \State {$b_r \gets (t_s + 1 + t_s + 1 - b_s) \% (t_r + 1)$}
    \State {$b_s \gets 0$}
    \State{Reverse to the receiving state by delaying $(t_s + 1 - b_s)$ slot(s)}
\EndProcedure
\Procedure{Receive}{$queue$}
    \State{$cycle\_wait \gets 0$, $flag\_receive \gets \textrm{false}$}
    \While{true}
        \State{Listen and receive message $msg$}
        \If{$msg.id\_next == self.id$}
            \State{$flag\_receive \gets \textrm{true}$}
            \State{$queue.\textrm{push}(msg)$ and send ACK}
            \State{$b_{r} = (t_s + 1 + b_r) \% (t_r + 1)$}
            \If{$msg.is\_last == \textrm{true}$}
                \State{$flag\_receive \gets \textrm{false}$}
                \State{Break}
            \EndIf
        \EndIf
        \If{$flag\_receive == \textrm{false}$}
            \State{$cycle\_wait \gets cycle\_wait + 1$}
            \If{$cycle\_wait > max\_wait$}
                \If{$queue.size == 0$}
                    \State{Continue}
                \EndIf
                \State{Break}
            \EndIf
        \EndIf
    \EndWhile
    \State{$b_s \gets (t_r + 1 - b_r) \% (t_s + 1)$}
    \State{$b_r \gets 0$}
    \State{Reverse to the sending state by delaying $t_r + 1 - b_r$ slot(s)}
\EndProcedure
\end{algorithmic}
\end{algorithm}

\subsubsection{Sending Procedure} 
The top of Algorithm~\ref{algo:forward} lists the sending procedure \textsc{Send($queue$)} of the \ours-forward protocol. For a given node, we assume it is already in sync with its next hop and has entered the Forward state. If the queue (for caching messages to be sent out) is not empty, the node maintains the synchronized state with its next hop (lines 5--6). In every charging cycle, the node pops a message from the queue, sets the metadata of the message, and sends the message out (lines 7--8). If it times out before receiving the ACK (e.g., due to collisions), it will keep retrying until the maximum number of attempts reaches an upper threshold and the \textsc{Send} procedure aborts directly (lines 9--14). Meanwhile, it updates $b_s$ as we explained before (line 15). When the queue becomes empty, the sending procedure is terminated and the node reverses its state and starts to follow the receiving procedure (lines 17--19). The node also resets $b_r$ and $b_s$ as explained in the Reverse state (lines 21--23). In case the next hop cannot be reached after a few charging cycles, the sender concludes that the charging cycle of the sender or receiver, or both has changed and calls \sync to re-synchronize with its next hop.

\subsubsection{Receiving Procedure}
The bottom of Algorithm~\ref{algo:forward} lists the receiving procedure \textsc{Receive($queue$)} of the \ours-forward protocol. If the node is synchronized with a previous hop, the node listens and receives incoming messages (line 28). If the incoming message is valid, the node pushes the message to its queue and sends an acknowledgment back (lines 29--31). The node also updates its working period bias with respect to its own next hop $b_r$, where in every charging cycle, the node advances $(t_s + 1 + b_r)\%(t_r + 1)$ slot(s) on the charging cycle of its next hop (line 32). Upon the last message from the synchronized previous hop, the node leaves the connection (lines 33--36) and then reverses to the sending procedure directly. If the expected receiving is not successful, the node counts down its waiting cycles. When the number of waiting cycles reaches the maximum (empirically set to $\alpha(t_s + 1)$) (lines 40--45), the node performs the following depending on the queue status: If there are no messages in its queue, the node keeps waiting (line 42); otherwise, the node switches to the sending state (line 44). As a result, the node applies Sync-next to synchronize its working time with its next hop and resets $b_s$ and $b_r$ as explained in the Sync-next state (lines 48--50).

\section{Route Construction}
\label{sec:route}
\begin{algorithm}[!t]
\small
\caption{Route construction protocol \ours-route}\label{algo:route}
\begin{algorithmic}[1]
\State{\textbf{Input:} node's id: $id\_self$, node's charging time and its upper-bound: $t\_self$, $t\_max$, maximum coprime gap: $\delta$}
\State{\textbf{Output:} node's next hop: $id\_next$}
\State{$cur\_hop \gets \infty$, $t\_wait \gets 0$, $id\_next \gets \textrm{null}$, $c\_inc \gets 0$}
\Procedure{MsgHandler}{$msg$}
    \If{$cur\_hop > msg.hop + 1$}
        \State{$id\_next \gets msg.id$}
        \State{$cur\_hop \gets msg.hop + 1$}
        \State{$t\_wait \gets msg.t\_wait$}
    \EndIf
    \While{$t\_wait > 0$}
        \State{$t\_wait \gets t\_wait - 1$}
        \State{Call \textsc{MsgHandler()} upon new messages}
    \EndWhile
    \State{\textsc{Broadcast}($id\_self$, $cur\_hop$)}
\EndProcedure

\Procedure{Broadcast}{$id\_self$, $cur\_hop$}
    \MyComment{Maximum attempts to ensure all neighbors are reached}
    \While{$c\_inc < \delta$}
        \State{$cycle\_count \gets 0$}
        \While{$cycle\_count < t\_max$}
            \State{$cycle\_count \gets cycle\_count + 1$}
            \State{$t\_wait \gets t\_max(\delta-c\_inc) - cycle\_count$} 
            \State{Delay working period by $c\_inc$}
            \State{Broadcast $\langle id\_self,cur\_hop, t\_wait \rangle$}
        \EndWhile
        \State{$c\_inc \gets c\_inc + 1$}
    \EndWhile 
\EndProcedure
\end{algorithmic}
\end{algorithm}

We propose a route construction protocol \route, which is simple to implement and ensures each node reaches the sink with the minimum hops. The workings of \route are listed in Algorithm~\ref{algo:route}. Each node runs two procedures: \textsc{MsgHandler} and \textsc{Broadcast}. The \textsc{MsgHandler} procedure is responsible for processing incoming messages. Upon receiving a message, we check if a shorter path has been found with the received message. If so, we keep the shorter path and update the hop count (lines 5--9) and wait for a certain amount of time (i.e., $t\_wait$ cycles of $t\_max$ slots each) announced by the message source. During the waiting period, if another message with an even shorter path is received, we update the hop count and restart the waiting timer with the waiting time contained in the new message (lines 10--13). Upon timeout, the node broadcasts its hop count with the \textsc{Broadcast} procedure.

The \textsc{Broadcast} procedure performs maximum attempts to ensure that the message of a node can reach all the node's potential neighbors (lines 18--27). We assume the maximum charging period is $t\_max$ long and hence, the maximum number of rounds is bounded by $\delta \cdot t\_max$ where $\delta$ is the maximum coprime gap. The waiting time necessary for the message receiver to wait before they broadcast is the time left for the current node to finish its broadcast (line 22). This is to avoid a downstream node starting to broadcast its hop count before it hears from all possible upstream nodes. In every charging cycle, the node delays its working period by $c\_inc$ and sends out the hop count message together with the expected waiting time, hoping to reach a potential neighbor. $c\_inc$ is incremented until it reaches $\delta$ to ensure a coprime to the neighbor's charging cycle is covered for sure. Therefore, the following theorem holds.

\begin{theorem}
    \route ensures that the hop count of the path from each node to the sink is minimized.
\end{theorem}
\begin{proof}
    We define the height of a node as the least hop count to the sink. Since \route always keeps the path with the least hop count upon receiving messages, the proof can be conducted by showing that a node can't choose its next hop with the same height. We consider the contradiction scenario shown in Figure~\ref{fig:route:proof}. 
    Let us assume hop D-F is part of the path of a node. This is only possible if node F has received the broadcast from node D and has started its own broadcasting before receiving the broadcast from C (hence it missed the broadcast message from C). This means that node F must have waited for D to finish its broadcasting and node D must have also waited for B to finish its broadcasting. Since in our route construction protocol, we set the maximum broadcast time of a node to the bound $t\_bcast = (t\_max + 1)\cdot \delta \cdot t\_max$ slots. This is based on the worst case where the node has the maximum charging time in the system. Hence, node F has to wait at least $2\times t\_bcast$ slots after B starts to broadcast before it can start broadcasting. However, with $2\times t\_bcast$, node C is guaranteed to have finished its broadcasting (since B and C are at the same height and thus require the same amount of time for finishing broadcasting), which contradicts that node F has not received the broadcast from C. This completes the proof. 
\end{proof}

\begin{figure}[!t]
    \centering
    \includegraphics[page=2,width=0.3\textwidth]{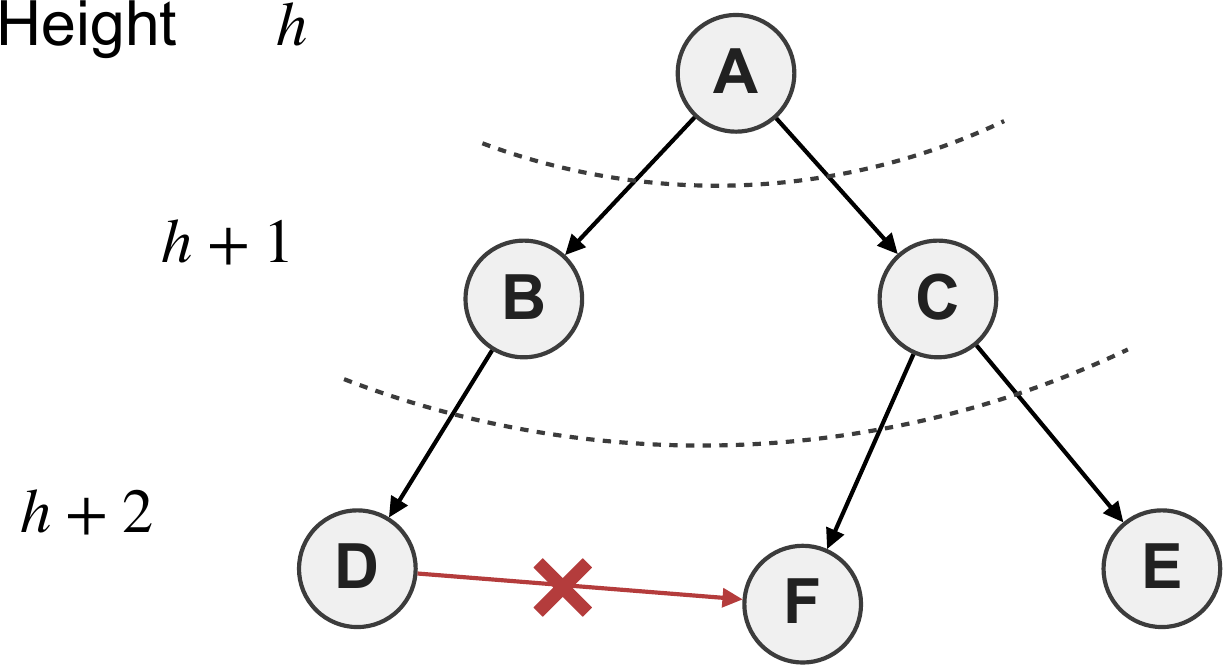}
    \caption{Illustration for a contradiction where we show that hop D-F cannot exist on any constructed route.}
    \label{fig:route:proof}
\end{figure}

\section{Performance Evaluation}
\label{sec:eval}
We implement \ours in OMNeT++ and perform extensive experiments with a variety of real-world parameter setups and compare \ours with state-of-the-art approaches retrofitted with intermittency support. We choose OMNeT++ since it allows us to focus on the key aspects of the protocol and assess the performance on large scales. In the following, we first explain our evaluation setups.
For \ours-sync we also implement a prototype and test it on a hardware testbed.
Then, we compare \ours with baselines and study the impact of different environmental factors.

\subsection{Experimental Setups}
Following findings in existing works~\cite{2021-nsdi-find,2022-nsdi-bonito}, we assume the charging time $t$ of all  nodes falls into the range of $[5, 500]$ time slots where a time slot is set to 1ms. Additionally, for two nodes $i$ and $j$, we assume their charging time difference is bounded by a small constant as mentioned before, i.e., $t_j/\gamma \leq t_i \leq \gamma t_j$ where $\gamma \in \{2,3,4\}$. We consider three ranges of node charging time, $[5,15]$ for good energy conditions, $[40,120]$ for medium energy conditions, and $[166,498]$ for poor energy conditions. For evaluating the node synchronization protocol, we use the public dataset for real-world charging times of BF devices~\cite{2021-nsdi-find} and set $\gamma$ to three. 

We consider two cases for the sensing system configuration: \emph{square} and \emph{rectangle}. In the square case, 100 nodes are randomly distributed in a square area of $45$ m $\times$ $45$ m, while the same number of nodes are randomly distributed in a rectangle area of $25$ m $\times$ $100$ m in the rectangle case. In both cases, the sink node sits in the middle of the area's right edge. These two cases show how \ours performs under different environments with a varying average number of hops from each node to the sink. We also vary BF devices' communication range (10m, 15m, 20m) to analyze its impact on the routing performance. Our baselines are chosen according to the evaluation goals, as detailed in the following subsections.

\subsection{Performance of \sync}
We first evaluate the performance of \sync by comparing it with Find~\cite{2021-nsdi-find}, 
which is the state-of-the-art node synchronization mechanism between BF devices. We adapt Find to scenarios with heterogeneous energy conditions. We consider two nodes to be synchronized. With Find nodes delay for a random number of slots following a geometric distribution with fine-tuned parameters. We consider all three energy conditions (i.e., good, medium, and poor) where we filter out the charging time pairs (each for one of the nodes) provided in a public dataset~\cite{2021-nsdi-find} falling into these energy conditions. Due to the randomization of Find, for each pair of charging times, we repeat simulations 1000 times to avoid noise. 

\begin{figure*}[!t]
    \centering
    \includegraphics[width=0.8\textwidth]{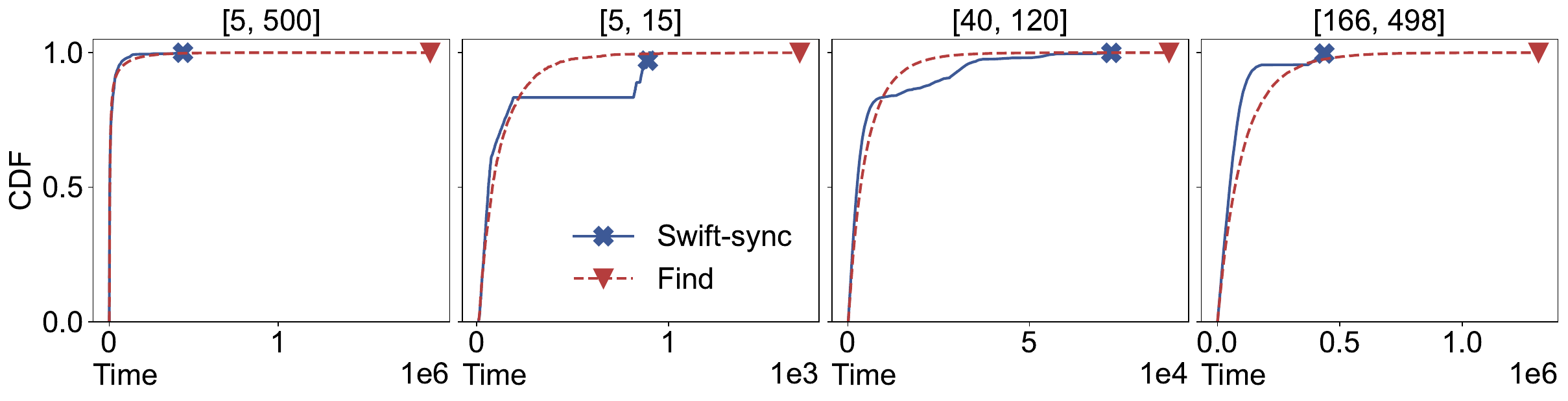}
    \caption{Node synchronization time (in slots) with \sync and Find~\cite{2021-nsdi-find} under varying energy condition ranges. The markers denote the 99th percentile.}
    \label{fig:eval:sync}
\end{figure*}

Figure~\ref{fig:eval:sync} depicts the cumulative distribution function (CDF) of the node synchronization time (in slots) with \sync and Find under the three energy conditions as well as the aggregated results. We observe that \sync outperforms Find in three aspects: (1) Under all three energy conditions, \sync achieves lower synchronization time than Find at around the 80th percentile. We can also observe the long tail of the latency CDF of Find, demonstrating the uncertainty caused by randomness. 
(2) Under good and median energy conditions (ranges $[5, 15]$ and $[40, 120]$), \sync outperforms Find under around the 80th percentile, while \sync is better than Find under around 98th percentile under poor energy conditions (range $[166,498]$). Compared with Find, \sync reduces the average synchronization time by more than $20\%$. 

\begin{figure}[!t]
    \centering
    \includegraphics[width=0.3\textwidth]{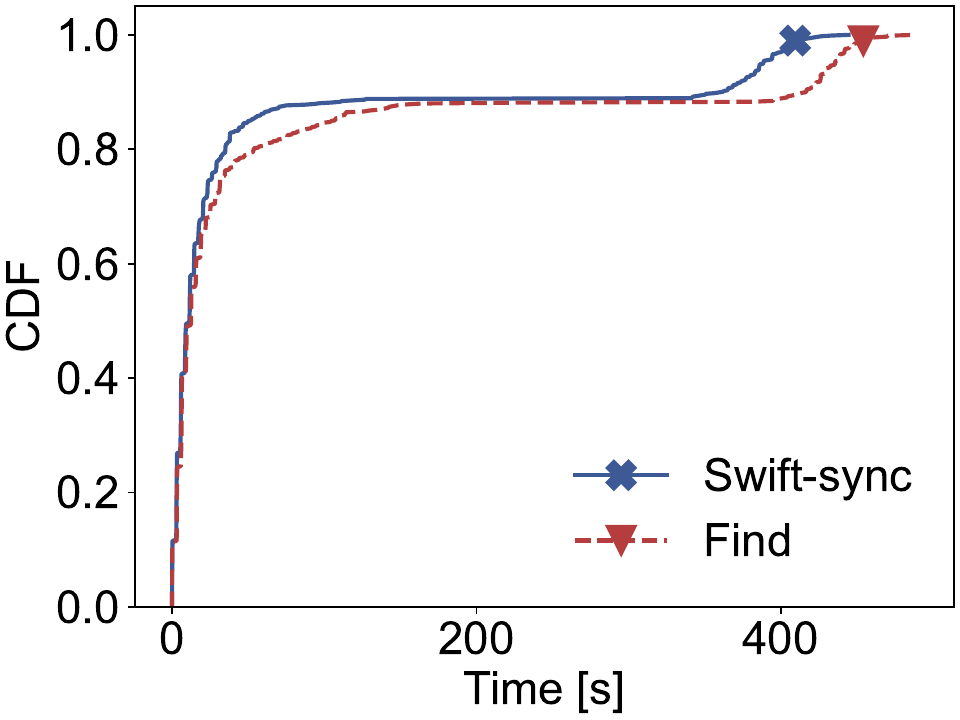}
    \caption{Node synchronization time on the hardware prototype.}
    \label{fig:syn:testbed}
\end{figure}

To further validate \ours-sync, we develop a prototype with two intermittently powered BF devices. The BF device consists of three parts: a control unit, a power unit, and a communication unit. Both the control and power units are built with TI-MSP430FR5994 boards, while the communication unit is built with a NORDIC-NRF52840 board for BLE communication. The power and control units are interconnected with UART and the control unit talks to the BLE unit through SPI. The power unit controls a MOSFET to supply power to the control unit, emulating different intermittency behaviors in the power supply. With this setup, we tested the performance of \ours-sync under the charging time range $[40,120]$ slots where a slot lasts for 41 ms based on our hardware setup. The much higher slot length is due to the overheads incurred by the communication between multiple boards to emulate an intermittently-powered BF device. 
Nonetheless, our protocol still applies.  Figure~\ref{fig:syn:testbed} shows that the result matches that from the simulation. We can see that \ours-sync outperforms Find in most cases. Overall, \sync achieves efficient synchronization while ensuring reliability.

\begin{figure}[!t]
    \centering
    \includegraphics[width=0.48\textwidth]{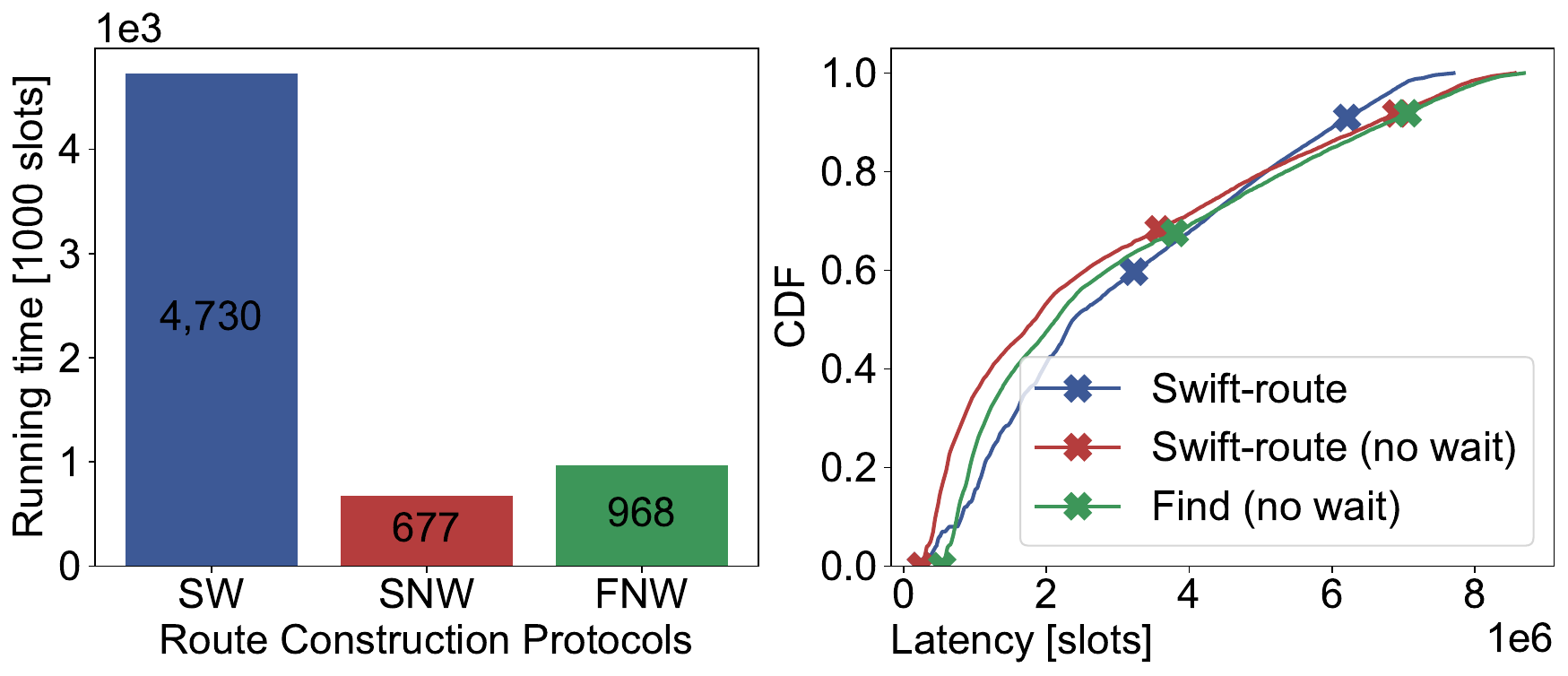}
    \caption{(left) route construction time (in slots) of \route with waiting (SW), \route without waiting (SNW), and Find without waiting (FNW), and (right) end-to-end message delivery time CDF of the three protocols.}
    \label{fig:eval:route}
\end{figure}

\subsection{Performance of \route}
\begin{figure*}[!t]
    \centering
    \subfloat[Square 45m $\times$ 45m]{
        \includegraphics[width=0.8\textwidth]{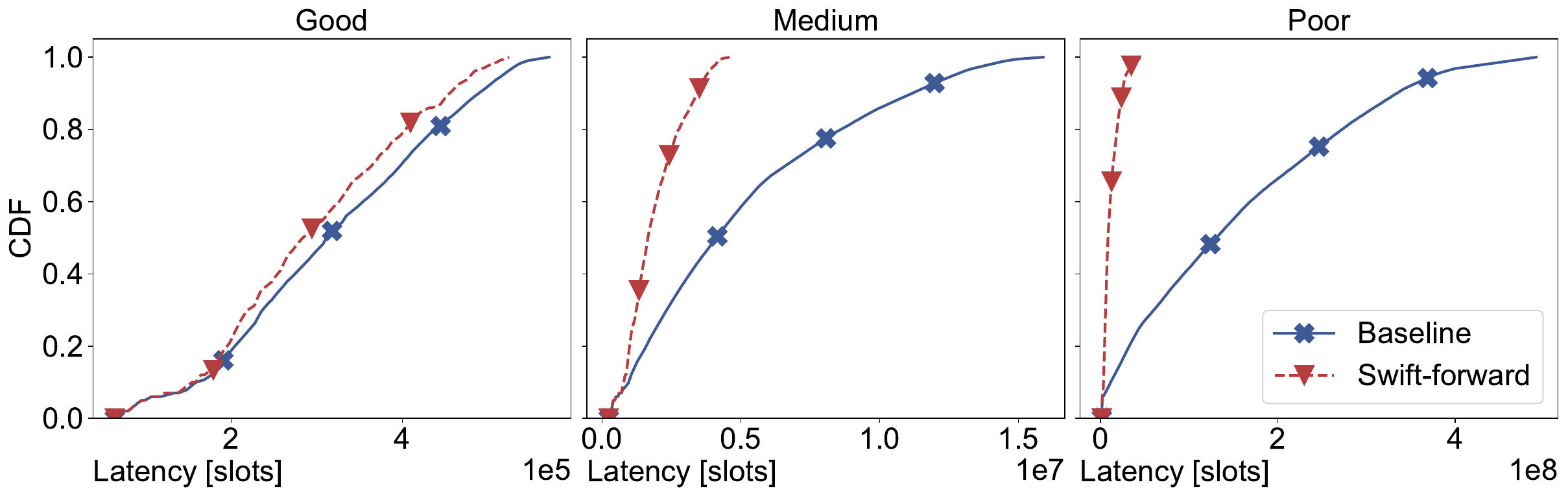}}
    \hfill
    \subfloat[Rectangle 25m $\times$ 100m]{
        \includegraphics[width=0.8\textwidth]{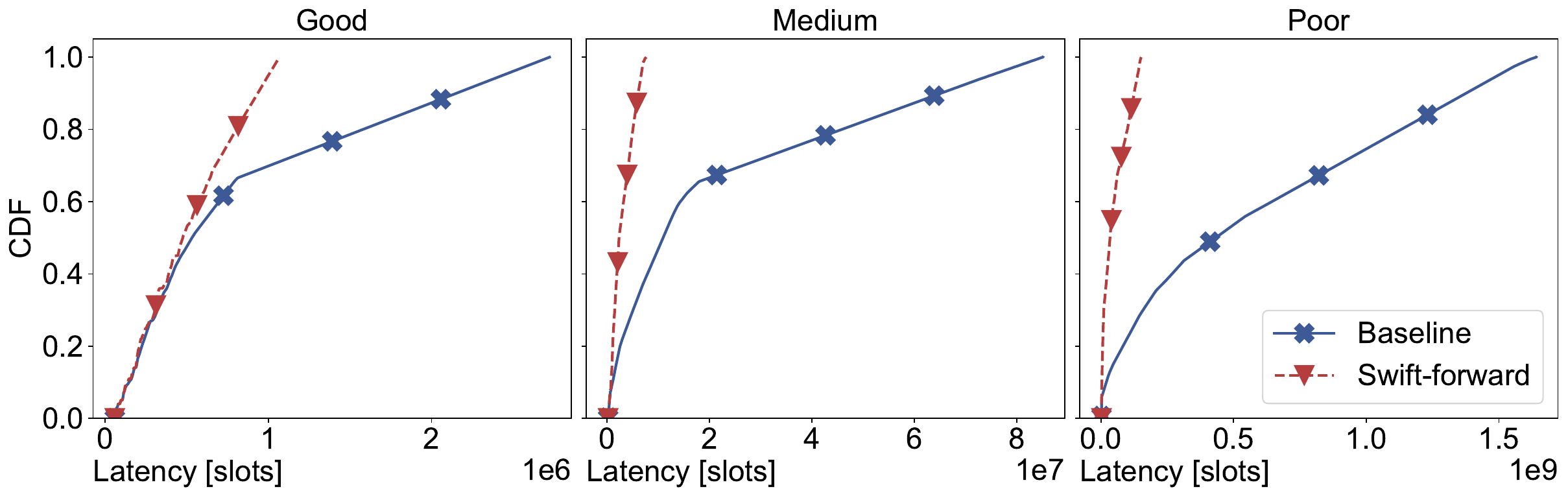}}
    \caption{Message forwarding time (in slots) comparison between \forward and a baseline without fast synchronization under varying energy conditions and two field scenarios.}
    \label{fig:eval:forward:overall}
\end{figure*}

\begin{figure*}[!t]
    \centering
\includegraphics[width=0.85\textwidth]{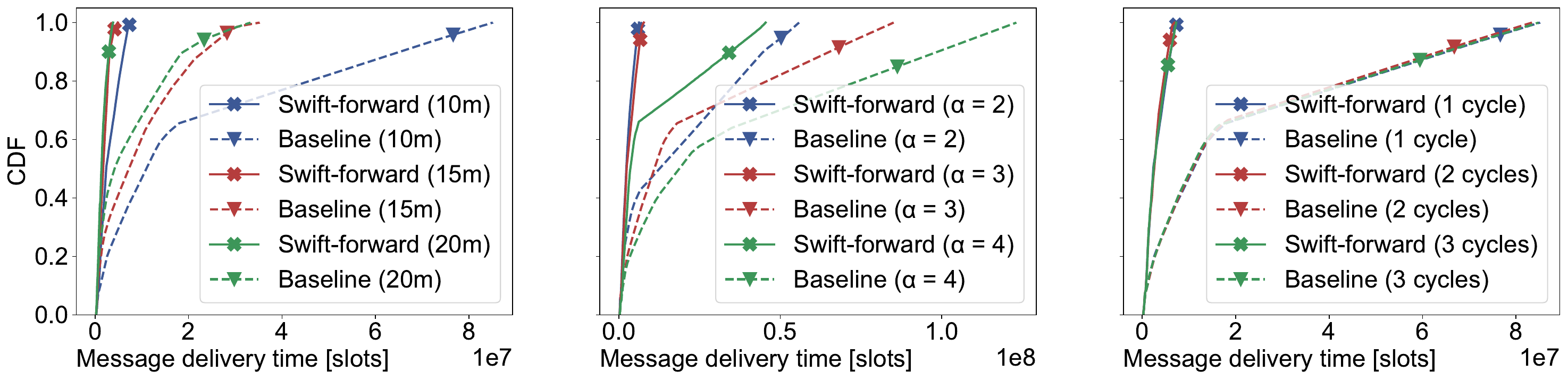}
    \caption{Performance of \ours under varying parameter setups: (left figure) the communication range, (middle figure) the maximum charging time difference, and (right figure) the message intensity.}
    \label{fig:eval:forward:impact}
\end{figure*}

We now evaluate the performance of \route. A major feature of \route is its optimality regarding the hop count of each node to the sink. This is achieved by the waiting mechanism in the design of the protocol. We thus compare with a variant of \route where this waiting mechanism is disabled, i.e., each node broadcasts its hop count immediately after they receive a message containing a smaller hop count. We call this variant \route (no-wait). We also compare with a baseline where the node synchronization is done with Find, without the waiting mechanism. We call this baseline Find (no-wait). Traditional routing protocols for WSNs are not considered due to their inability to handle intermittency.

Figure~\ref{fig:eval:route}(left) shows the time (in slots) required to complete the route construction protocol on the rectangle field under the medium energy condition. As expected, \route takes the longest time, due to the waiting mechanism in the protocol to ensure optimality. Considering that the protocol is only needed to run to establish the routes at the beginning, this time overhead is acceptable. While taking far less time to complete, the baselines produce suboptimal paths, leading to long delays in message delivery. Figure~\ref{fig:eval:route}(right) shows the message delivery time CDF when using the routes generated by each of the three approaches. We can see that \route can largely reduce the message delivery time for nodes with long routes compared with the baselines thanks to the shorter paths for these nodes.
However, this is achieved at the expense of higher message delivery time for nodes with short routes between them. One major reason is that the shortest paths produced by \ours form the shortest path tree, with all other links not on any shortest paths pruned out, leading to the worst path diversity. This is likely to create more congestion on some links compared with the baselines where different links may be used to spread out the load.

\subsection{Performance of \forward}
To evaluate \forward, we compare it with a baseline where we disable fast synchronization in \forward. In other words, the baseline employs the node synchronization mechanism every time a node communicates with its next hop. We evaluate all three energy conditions (good, medium, poor) and with the two areas (square and rectangle).

Figure~\ref{fig:eval:forward:overall} shows the CDF results for the end-to-end message delivery time (in slots) under the above setups. As we can observe, \forward is highly efficient in message forwarding, where the end-to-end message delivery time is reduced by up to an order of magnitude under all the considered conditions. Comparing the results for the different energy conditions (good, medium, and poor), we see a clear trend where the gap between \forward and the baseline becomes larger when the energy condition improves, suggesting that \forward can maintain its efficiency even in hostile environments.

\subsection{Impact of Parameters}
We also test \forward under different setups to examine the impact of system parameters on the message forwarding performance. We focus on the end-to-end message delivery time (in slots) as the performance metric and choose the same baseline (without fast synchronization) as we used in the previous experiment. 

Figure~\ref{fig:eval:forward:impact} shows the performance of \forward under varying parameter setups. The left figure depicts the \emph{impact of the communication range} where we consider each node can communicate directly with other nodes within a range of 10m, 15m, and 20m, respectively. We can see that \forward performs better than the baseline in all cases. Moreover, when the communication range increases, the end-to-end message delivery time decreases as expected, which is attributed to the reduced hop count for each node to reach the sink. The middle figure depicts the \emph{impact of the maximum charging time difference} (i.e., $\alpha$), where we show three cases with $\alpha=2$, $\alpha=3$, and $\alpha=4$, respectively. Similarly, \forward outperforms the baseline in all cases. With the maximum charging time difference increase, the end-to-end message delivery time generally increases. This can be explained by the fact that the node synchronization takes more time with a larger charging time difference between nodes. Finally, the right figure shows the \emph{impact of the message intensity} in the network. We consider three message intensity levels where each node sends a message every one, two, and three cycles, respectively. Again, \forward beats the baseline in all cases as expected. We do not observe much performance difference in these cases, which indicates that \forward is equally efficient under all the considered message intensities. 

The selection of these parameters highly depends on the specific application requirements and deployment conditions. In general, it is always beneficial to increase the communication range of BF devices, if the BF devices are equipped with capacitors large enough to sustain the communication in the working period of the device. The charging time range is usually dictated by the environment. However, it is generally beneficial to avoid deploying BF devices at locations with extremely bad energy availability to achieve an $\alpha$ as small as possible.

\section{Related works}
\label{sec:relatedwork}
\textbf{Communication for BF devices.} 
Communication between BF devices has gained significant attention, and the first challenge is addressing the node synchronization problem. Without shared knowledge, node synchronization is attempted by delaying the working period of one or both devices following specific distributions~\cite{2018-tecs-predic,2021-mass-comm,2021-nsdi-find}, hoping they align after several rounds. These approaches require fine-tuning based on system parameters and energy conditions and guarantee no success due to their randomized nature. Bonito proposes to learn the charging time with well-known distributions to enable continuous communication~\cite{2022-nsdi-bonito}. De Winkel et al. also present a Bluetooth stack for battery-free devices without changing the protocol specification~\cite{2022-mobisys-bluetooth}. While focusing on a single hop, these works provide a solid foundation for the routing problem we study in this paper.
In short, existing routing protocols fall short because of (1) unrealistic assumptions not holding in intermittent sensing systems, making them inapplicable, and (2) not accounting for the high synchronization overhead in intermittent sensing systems, making them inefficient.

\textbf{Low-power and lossy wireless networks.} 
Significant research has been conducted on low-power and lossy wireless networks (LLWNs), including time-slotted channel hopping~(TSCH)~\cite{2019-nca-tsch-survey}  
and routing protocols for low power and lossy networks~(RPL)~\cite{2019-sensors-llns}. Within this domain, concepts such as frame-slots and sub-slots have been introduced. Notably, sub-slots are subsets of frame-slots, typically spanning a 10-millisecond scale~\cite{2021-ieee-tsch, 2011-acm-ch, 2014-ICS-iscc, 2020-infocom-ost}. However, this does not imply that sensor nodes operate within this time frame. Instead, coordinators utilize these time slots to schedule communications, effectively reducing collision risks among sensor nodes. In contrast, intermittently-powered networks present unique challenges. Here, the active working time of BF devices is measured at a millisecond scale. Predicting their charging time is complex, and for this study, we consider a simplified scenario. Moreover, unlike traditional networks where a coordinator manages synchronization, in intermittently-powered sensing systems, BF devices independently handle their synchronization processes.

\textbf{Routing in sensing systems.}
Routing is an essential function of any sensing system and has been heavily studied in WSNs. Assuming continuous power supply by the onboard battery, traditional WSNs often leverage potent devices that are not available for intermittently-powered BF sensing systems~\cite{2020-algo-wsnsurvey, 2021-csur-sync}. Meanwhile, energy-harvesting WSNs employ energy harvesters to recharge the onboard battery or super-capacitors and adopt active duty cycling to control the device state intentionally~\cite{2018-tsn-ehwsn,2009-icnp-esc, 2011-mobihoc-esa ,2016-infocom-panda,2018-infocom-ehwsn,2014-mobihoc-eh,2019-sigcomm-nd, 2022-tmc-das, 2011-mobihoc-esa, 2016-infocom-ehdataprocess}.
Both traditional WSNs and energy-harvesting WSNs assume more powerful energy storage than intermittently-powered sensing systems.

\textbf{Linear congruential generator.} 
The Linear Congruential Generator (LCG) is a widely recognized method for generating random numbers and has been a staple in the realm of computational algorithms. LCG plays a crucial role in security communications and is frequently employed in wireless sensor networks for random number generation. To our knowledge, LCG has been utilized for encrypting data in several studies, as indicated in references~\cite{2023-mtp-smartFa, 2008-cc-rfid, 2020-sensors-smartAg}. Additionally, EDGF~\cite{2021-cc-edfg} has leveraged LCG for generating datasets in simulation environments. Notably, Swift represents the pioneering work in using a customized LCG for synchronizing BF devices, representing a significant advancement in this field.

\section{Discussion}
\label{sec:discussion}
\subsection{Comparative Analysis with EH-WSNs}
Traditional wireless sensor networks (WSNs) typically consist of sensor nodes that are powered by chemical batteries~\cite{2020-algo-wsnsurvey}. The routing problem has been studied extensively, but the goal is mainly to pursue QoS and the whole network life through balancing the energy consumption of nodes. In light of the maintenance challenge, researchers have explored the idea of energy-harvesting WSNs (EH-WSNs), where the sensor node is powered by a rechargeable battery or a supercapacitor~\cite{2018-tosn-ehwsn}. With this difference, the goal of routing optimization turns to the QoS and the uptime of the network. The uptime of the network depends on the energy availability which is dictated by the energy harvesting rate and the battery/capacitor size. Most of the existing works leverage prediction models to assess the energy availability and perform energy management, based on which the routing schemes are developed~\cite{2018-tosn-ehwsn,2015-commag-ehwsn,2019-ipsn-ehs}. Intermittently-powered battery-free sensing systems are fundamentally different from EH-WSNs since the buffered energy in the capacitor will be depleted in just one wake-up of the BF device, lasting only one or a few milliseconds. Also, the charging time is much shorter due to the limited capacity of the tiny capacitor.
Therefore, there is no energy management needed across device wake-ups unlike in EH-WSNs. Our routing scheme is specially tailored for intermittently powered BF sensing systems, with the goal of minimizing the number of wake-ups needed for device synchronization and message passing.

\subsection{Scalability}
Scaling intermittently powered sensing systems remains a challenge in general due to the high level of dynamism and unpredictability the environment may impose on the system. While taking a significant step, \ours is still limited by the assumption it builds on: the charging time of a BF device stays stable at a small time scale (e.g., a few minutes). This assumption typically holds for deployment in indoor environments (e.g., a greenhouse or warehouse) where the environment is highly controlled and more stable than outdoor wild environments. For applications in a greenhouse or warehouse, we think that the size we are targeting (i.e., 100 BF devices used in the evaluation) is reasonable. For indoor environments on even larger scales, Swift can still work, albeit the performance (w.r.t. node synchronization time and end-to-end message delivery time) may degrade inevitably due to the increase of the hop count in the topology. One mitigation strategy is to increase the communication range of BF devices to reduce the hop count. However, this will require equipping these BF devices with larger capacitors, thus negatively affecting the form factor of the device as well as the deployment cost. For more hostile conditions where the charging time of BF devices varies dramatically in an unpredictable fashion, Swift will unlikely work as expected, and future research work is needed.

\subsection{Energy Efficiency}
Energy efficiency is an important topic for embedded systems like the sensing system we consider in this paper. Nevertheless, we would like to point out that there are fundamental differences between intermittently powered battery-free devices and traditional battery-powered sensor devices. With traditional battery-powered sensor devices, the energy saved from an efficient protocol design can be used for other tasks of the system. However, intermittently powered battery-free devices wake up when enough energy has been collected and stored in the onboard capacitor with an artificial delay imposed by the protocol. In every device wake-up, the device only performs a fixed communication round. As a result, the energy efficiency of the protocol can mainly be characterized by the number of device wake-ups needed for synchronization and message passing.
This means that the performance metrics we focus on in the evaluation, e.g., the synchronization time, the route construction latency, and the end-to-end message delivery time, already capture the energy efficiency aspect of the protocol. In short, the faster the protocol is, the more energy efficient the protocol is, which means that fewer charging cycles will be spent on the protocol and more cycles on other tasks.

\section{Conclusions}
\label{sec:conclusions}
This paper studies routing with intermittency and proposes \ours---a routing scheme tailored for intermittently-powered BF sensing systems. \ours features three new designs: a node synchronization protocol \sync enabling the communication between BF devices with success guarantee, a message forwarding protocol \forward that avoids repeated synchronization between nodes on the route to achieve low latency, and a route construction mechanism \route that produces a route from each node to the sink node with the least hop count. Our evaluation based on an implementation in OMNeT++ and real-world setups shows that \ours achieves efficient routing and outperforms the state-of-the-art to a large extent. In future work, we aim to explore scenarios involving higher temporal variations in charging time.

\bibliographystyle{IEEEtran}
\bibliography{refs}

\begin{thebibliography}{10}
\providecommand{\url}[1]{#1}
\csname url@samestyle\endcsname
\providecommand{\newblock}{\relax}
\providecommand{\bibinfo}[2]{#2}
\providecommand{\BIBentrySTDinterwordspacing}{\spaceskip=0pt\relax}
\providecommand{\BIBentryALTinterwordstretchfactor}{4}
\providecommand{\BIBentryALTinterwordspacing}{\spaceskip=\fontdimen2\font plus
\BIBentryALTinterwordstretchfactor\fontdimen3\font minus
  \fontdimen4\font\relax}
\providecommand{\BIBforeignlanguage}[2]{{%
\expandafter\ifx\csname l@#1\endcsname\relax
\typeout{** WARNING: IEEEtran.bst: No hyphenation pattern has been}%
\typeout{** loaded for the language `#1'. Using the pattern for}%
\typeout{** the default language instead.}%
\else
\language=\csname l@#1\endcsname
\fi
#2}}
\providecommand{\BIBdecl}{\relax}
\BIBdecl

\bibitem{2017-sensys-futureic}
\BIBentryALTinterwordspacing
J.~D. Hester and J.~Sorber, ``The future of sensing is batteryless,
  intermittent, and awesome,'' in \emph{Proceedings of the 15th {ACM}
  Conference on Embedded Network Sensor Systems, SenSys 2017, Delft,
  Netherlands, November 06-08, 2017}, M.~R. Eskicioglu, Ed.\hskip 1em plus
  0.5em minus 0.4em\relax {ACM}, 2017, pp. 21:1--21:6. [Online]. Available:
  \url{https://doi.org/10.1145/3131672.3131699}
\BIBentrySTDinterwordspacing

\bibitem{2017-snapl-survey}
\BIBentryALTinterwordspacing
B.~Lucia, V.~Balaji, A.~Colin, K.~Maeng, and E.~Ruppel, ``Intermittent
  computing: Challenges and opportunities,'' in \emph{2nd Summit on Advances in
  Programming Languages, {SNAPL} 2017, May 7-10, 2017, Asilomar, CA, {USA}},
  ser. LIPIcs, B.~S. Lerner, R.~Bod{\'{\i}}k, and S.~Krishnamurthi, Eds.,
  vol.~71.\hskip 1em plus 0.5em minus 0.4em\relax Schloss Dagstuhl -
  Leibniz-Zentrum f{\"{u}}r Informatik, 2017, pp. 8:1--8:14. [Online].
  Available: \url{https://doi.org/10.4230/LIPIcs.SNAPL.2017.8}
\BIBentrySTDinterwordspacing

\bibitem{2022-tecs-camaroptera}
\BIBentryALTinterwordspacing
H.~Desai, M.~Nardello, D.~Brunelli, and B.~Lucia, ``Camaroptera: {A} long-range
  image sensor with local inference for remote sensing applications,''
  \emph{{ACM} Trans. Embed. Comput. Syst.}, vol.~21, no.~3, pp. 32:1--32:25,
  2022. [Online]. Available: \url{https://doi.org/10.1145/3510850}
\BIBentrySTDinterwordspacing

\bibitem{2018-mobicom-eyetracker}
\BIBentryALTinterwordspacing
T.~Li and X.~Zhou, ``Battery-free eye tracker on glasses,'' in
  \emph{Proceedings of the 10th on Wireless of the Students, by the Students,
  and for the Students Workshop, S3@MobiCom 2018, New Delhi, India, November 2,
  2018}, S.~Pradhan and S.~K. Saha, Eds.\hskip 1em plus 0.5em minus 0.4em\relax
  {ACM}, 2018, pp. 27--29. [Online]. Available:
  \url{https://doi.org/10.1145/3264877.3264885}
\BIBentrySTDinterwordspacing

\bibitem{2021-weee-cic}
\BIBentryALTinterwordspacing
G.~Liu and L.~Wang, ``Self-sustainable cyber-physical systems with
  collaborative intermittent computing,'' in \emph{e-Energy '21: The Twelfth
  {ACM} International Conference on Future Energy Systems, Virtual Event,
  Torino, Italy, 28 June - 2 July, 2021}, H.~de~Meer and M.~Meo, Eds.\hskip 1em
  plus 0.5em minus 0.4em\relax {ACM}, 2021, pp. 316--321. [Online]. Available:
  \url{https://doi.org/10.1145/3447555.3465324}
\BIBentrySTDinterwordspacing

\bibitem{2020-ipsn-cis}
\BIBentryALTinterwordspacing
A.~Y. Majid, P.~Schilder, and K.~Langendoen, ``Continuous sensing on
  intermittent power,'' in \emph{19th {ACM/IEEE} International Conference on
  Information Processing in Sensor Networks, {IPSN} 2020, Sydney, Australia,
  April 21-24, 2020}.\hskip 1em plus 0.5em minus 0.4em\relax {IEEE}, 2020, pp.
  181--192. [Online]. Available:
  \url{https://doi.org/10.1109/IPSN48710.2020.00-36}
\BIBentrySTDinterwordspacing

\bibitem{2020-sensys-bfsensing}
\BIBentryALTinterwordspacing
M.~Afanasov, N.~A. Bhatti, D.~Campagna, G.~Caslini, F.~M. Centonze, K.~Dolui,
  A.~Maioli, E.~Barone, M.~H. Alizai, J.~H. Siddiqui, and L.~Mottola,
  ``Battery-less zero-maintenance embedded sensing at the mithr{\ae}um of
  circus maximus,'' in \emph{SenSys '20: The 18th {ACM} Conference on Embedded
  Networked Sensor Systems, Virtual Event, Japan, November 16-19, 2020},
  J.~Nakazawa and P.~Huang, Eds.\hskip 1em plus 0.5em minus 0.4em\relax {ACM},
  2020, pp. 368--381. [Online]. Available:
  \url{https://doi.org/10.1145/3384419.3430722}
\BIBentrySTDinterwordspacing

\bibitem{2023-ipsn-riotee}
\BIBentryALTinterwordspacing
K.~Geissdoerfer, I.~Splitt, and M.~Zimmerling, ``Demo abstract: Building
  battery-free devices with riotee,'' in \emph{The 22nd International
  Conference on Information Processing in Sensor Networks, {IPSN} 2023, San
  Antonio, TX, USA, May 9-12, 2023}.\hskip 1em plus 0.5em minus 0.4em\relax
  {ACM}, 2023, pp. 354--355. [Online]. Available:
  \url{https://doi.org/10.1145/3583120.3589808}
\BIBentrySTDinterwordspacing

\bibitem{2017-imwut-bfphone}
\BIBentryALTinterwordspacing
V.~Talla, B.~Kellogg, S.~Gollakota, and J.~R. Smith, ``Battery-free
  cellphone,'' \emph{Proc. {ACM} Interact. Mob. Wearable Ubiquitous Technol.},
  vol.~1, no.~2, pp. 25:1--25:20, 2017. [Online]. Available:
  \url{https://doi.org/10.1145/3090090}
\BIBentrySTDinterwordspacing

\bibitem{2020-imwut-bfgame}
\BIBentryALTinterwordspacing
J.~de~Winkel, V.~Kortbeek, J.~D. Hester, and P.~Pawelczak, ``Battery-free game
  boy: Sustainable interactive devices,'' \emph{GetMobile Mob. Comput.
  Commun.}, vol.~25, no.~2, pp. 22--26, 2021. [Online]. Available:
  \url{https://doi.org/10.1145/3486880.3486888}
\BIBentrySTDinterwordspacing

\bibitem{2018-asplos-energystore}
\BIBentryALTinterwordspacing
A.~Colin, E.~Ruppel, and B.~Lucia, ``A reconfigurable energy storage
  architecture for energy-harvesting devices,'' in \emph{Proceedings of the
  Twenty-Third International Conference on Architectural Support for
  Programming Languages and Operating Systems, {ASPLOS} 2018, Williamsburg, VA,
  USA, March 24-28, 2018}, X.~Shen, J.~Tuck, R.~Bianchini, and V.~Sarkar,
  Eds.\hskip 1em plus 0.5em minus 0.4em\relax {ACM}, 2018, pp. 767--781.
  [Online]. Available: \url{https://doi.org/10.1145/3173162.3173210}
\BIBentrySTDinterwordspacing

\bibitem{2016-osdi-ratchet}
\BIBentryALTinterwordspacing
J.~van~der Woude and M.~Hicks, ``Intermittent computation without hardware
  support or programmer intervention,'' in \emph{12th {USENIX} Symposium on
  Operating Systems Design and Implementation, {OSDI} 2016, Savannah, GA, USA,
  November 2-4, 2016}, K.~Keeton and T.~Roscoe, Eds.\hskip 1em plus 0.5em minus
  0.4em\relax {USENIX} Association, 2016, pp. 17--32. [Online]. Available:
  \url{https://www.usenix.org/conference/osdi16/technical-sessions/presentation/vanderwoude}
\BIBentrySTDinterwordspacing

\bibitem{2016-oopsla-chain}
\BIBentryALTinterwordspacing
A.~Colin and B.~Lucia, ``Chain: tasks and channels for reliable intermittent
  programs,'' in \emph{Proceedings of the 2016 {ACM} {SIGPLAN} International
  Conference on Object-Oriented Programming, Systems, Languages, and
  Applications, {OOPSLA} 2016, part of {SPLASH} 2016, Amsterdam, The
  Netherlands, October 30 - November 4, 2016}, E.~Visser and Y.~Smaragdakis,
  Eds.\hskip 1em plus 0.5em minus 0.4em\relax {ACM}, 2016, pp. 514--530.
  [Online]. Available: \url{https://doi.org/10.1145/2983990.2983995}
\BIBentrySTDinterwordspacing

\bibitem{2018-sensys-ink}
\BIBentryALTinterwordspacing
K.~S. Yildirim, A.~Y. Majid, D.~Patoukas, K.~Schaper, P.~Pawelczak, and J.~D.
  Hester, ``Ink: Reactive kernel for tiny batteryless sensors,'' in
  \emph{Proceedings of the 16th {ACM} Conference on Embedded Networked Sensor
  Systems, SenSys 2018, Shenzhen, China, November 4-7, 2018}, G.~S.
  Ramachandran and B.~Krishnamachari, Eds.\hskip 1em plus 0.5em minus
  0.4em\relax {ACM}, 2018, pp. 41--53. [Online]. Available:
  \url{https://doi.org/10.1145/3274783.3274837}
\BIBentrySTDinterwordspacing

\bibitem{2020-sensys-state}
\BIBentryALTinterwordspacing
H.~A. Asad, E.~H. Wouters, N.~A. Bhatti, L.~Mottola, and T.~Voigt, ``On
  securing persistent state in intermittent computing,'' in \emph{Proceedings
  of the 8th International Workshop on Energy Harvesting {\&} Energy-Neutral
  Sensing Systems, ENSsys@SenSys 2020, Virtual Event, Japan, November 16,
  2020}.\hskip 1em plus 0.5em minus 0.4em\relax {ACM}, 2020, pp. 8--14.
  [Online]. Available: \url{https://doi.org/10.1145/3417308.3430267}
\BIBentrySTDinterwordspacing

\bibitem{2020-pacmpl-foundation}
\BIBentryALTinterwordspacing
M.~Surbatovich, B.~Lucia, and L.~Jia, ``Towards a formal foundation of
  intermittent computing,'' \emph{Proc. {ACM} Program. Lang.}, vol.~4, no.
  {OOPSLA}, pp. 163:1--163:31, 2020. [Online]. Available:
  \url{https://doi.org/10.1145/3428231}
\BIBentrySTDinterwordspacing

\bibitem{2020-pldi-scheduling}
\BIBentryALTinterwordspacing
K.~Maeng and B.~Lucia, ``Adaptive low-overhead scheduling for periodic and
  reactive intermittent execution,'' in \emph{Proceedings of the 41st {ACM}
  {SIGPLAN} International Conference on Programming Language Design and
  Implementation, {PLDI} 2020, London, UK, June 15-20, 2020}, A.~F. Donaldson
  and E.~Torlak, Eds.\hskip 1em plus 0.5em minus 0.4em\relax {ACM}, 2020, pp.
  1005--1021. [Online]. Available:
  \url{https://doi.org/10.1145/3385412.3385998}
\BIBentrySTDinterwordspacing

\bibitem{2020-tecs-checkpoint}
\BIBentryALTinterwordspacing
S.~Ahmed, N.~A. Bhatti, M.~H. Alizai, J.~H. Siddiqui, and L.~Mottola, ``Fast
  and energy-efficient state checkpointing for intermittent computing,''
  \emph{{ACM} Trans. Embed. Comput. Syst.}, vol.~19, no.~6, pp. 45:1--45:27,
  2020. [Online]. Available: \url{https://doi.org/10.1145/3391903}
\BIBentrySTDinterwordspacing

\bibitem{2020-asplos-tics}
\BIBentryALTinterwordspacing
V.~Kortbeek, K.~S. Yildirim, A.~Bakar, J.~Sorber, J.~D. Hester, and
  P.~Pawelczak, ``Time-sensitive intermittent computing meets legacy
  software,'' in \emph{{ASPLOS} '20: Architectural Support for Programming
  Languages and Operating Systems, Lausanne, Switzerland, March 16-20, 2020},
  J.~R. Larus, L.~Ceze, and K.~Strauss, Eds.\hskip 1em plus 0.5em minus
  0.4em\relax {ACM}, 2020, pp. 85--99. [Online]. Available:
  \url{https://doi.org/10.1145/3373376.3378476}
\BIBentrySTDinterwordspacing

\bibitem{2020-imwut-bfree}
\BIBentryALTinterwordspacing
V.~Kortbeek, A.~Bakar, S.~Cruz, K.~S. Yildirim, P.~Pawelczak, and J.~D. Hester,
  ``Bfree: Enabling battery-free sensor prototyping with python,'' \emph{Proc.
  {ACM} Interact. Mob. Wearable Ubiquitous Technol.}, vol.~4, no.~4, pp.
  135:1--135:39, 2020. [Online]. Available:
  \url{https://doi.org/10.1145/3432191}
\BIBentrySTDinterwordspacing

\bibitem{2022-pldi-wario}
\BIBentryALTinterwordspacing
V.~Kortbeek, S.~Ghosh, J.~D. Hester, S.~Campanoni, and P.~Pawelczak, ``Wario:
  efficient code generation for intermittent computing,'' in \emph{{PLDI} '22:
  43rd {ACM} {SIGPLAN} International Conference on Programming Language Design
  and Implementation, San Diego, CA, USA, June 13 - 17, 2022}, R.~Jhala and
  I.~Dillig, Eds.\hskip 1em plus 0.5em minus 0.4em\relax {ACM}, 2022, pp.
  777--791. [Online]. Available: \url{https://doi.org/10.1145/3519939.3523454}
\BIBentrySTDinterwordspacing

\bibitem{2022-mobisys-bluetooth}
\BIBentryALTinterwordspacing
J.~de~Winkel, H.~Tang, and P.~Pawelczak, ``Intermittently-powered bluetooth
  that works,'' in \emph{MobiSys '22: The 20th Annual International Conference
  on Mobile Systems, Applications and Services, Portland, Oregon, 27 June 2022
  - 1 July 2022}, N.~Bulusu, E.~Aryafar, A.~Balasubramanian, and J.~Song,
  Eds.\hskip 1em plus 0.5em minus 0.4em\relax {ACM}, 2022, pp. 287--301.
  [Online]. Available: \url{https://doi.org/10.1145/3498361.3538934}
\BIBentrySTDinterwordspacing

\bibitem{2021-nsdi-find}
\BIBentryALTinterwordspacing
K.~Geissdoerfer and M.~Zimmerling, ``Bootstrapping battery-free wireless
  networks: Efficient neighbor discovery and synchronization in the face of
  intermittency,'' in \emph{18th {USENIX} Symposium on Networked Systems Design
  and Implementation, {NSDI} 2021, April 12-14, 2021}, J.~Mickens and
  R.~Teixeira, Eds.\hskip 1em plus 0.5em minus 0.4em\relax {USENIX}
  Association, 2021, pp. 439--455. [Online]. Available:
  \url{https://www.usenix.org/conference/nsdi21/presentation/geissdoerfer}
\BIBentrySTDinterwordspacing

\bibitem{2022-nsdi-bonito}
\BIBentryALTinterwordspacing
------, ``Learning to communicate effectively between battery-free devices,''
  in \emph{19th {USENIX} Symposium on Networked Systems Design and
  Implementation, {NSDI} 2022, Renton, WA, USA, April 4-6, 2022},
  A.~Phanishayee and V.~Sekar, Eds.\hskip 1em plus 0.5em minus 0.4em\relax
  {USENIX} Association, 2022, pp. 419--435. [Online]. Available:
  \url{https://www.usenix.org/conference/nsdi22/presentation/geissdoerfer}
\BIBentrySTDinterwordspacing

\bibitem{2019-sigcomm-nd}
\BIBentryALTinterwordspacing
P.~H. Kindt and S.~Chakraborty, ``On optimal neighbor discovery,'' in
  \emph{Proceedings of the {ACM} Special Interest Group on Data Communication,
  {SIGCOMM} 2019, Beijing, China, August 19-23, 2019}, J.~Wu and W.~Hall,
  Eds.\hskip 1em plus 0.5em minus 0.4em\relax {ACM}, 2019, pp. 441--457.
  [Online]. Available: \url{https://doi.org/10.1145/3341302.3342067}
\BIBentrySTDinterwordspacing

\bibitem{2016-infocom-panda}
\BIBentryALTinterwordspacing
R.~Margolies, G.~Grebla, T.~Chen, D.~Rubenstein, and G.~Zussman, ``Panda:
  Neighbor discovery on a power harvesting budget,'' vol.~34, no.~12, 2016, pp.
  3606--3619. [Online]. Available:
  \url{https://doi.org/10.1109/JSAC.2016.2611984}
\BIBentrySTDinterwordspacing

\bibitem{2016-rtss-eh}
\BIBentryALTinterwordspacing
Z.~Dong, Y.~Gu, J.~Chen, S.~Tang, T.~He, and C.~Liu, ``Enabling predictable
  wireless data collection in severe energy harvesting environments,'' in
  \emph{2016 {IEEE} Real-Time Systems Symposium, {RTSS} 2016, Porto, Portugal,
  November 29 - December 2, 2016}.\hskip 1em plus 0.5em minus 0.4em\relax
  {IEEE} Computer Society, 2016, pp. 157--166. [Online]. Available:
  \url{https://doi.org/10.1109/RTSS.2016.024}
\BIBentrySTDinterwordspacing

\bibitem{2012-csur-wsn}
\BIBentryALTinterwordspacing
R.~A. Uthra and S.~V.~K. Raja, ``Qos routing in wireless sensor networks - a
  survey,'' \emph{{ACM} Comput. Surv.}, vol.~45, no.~1, pp. 9:1--9:12, 2012.
  [Online]. Available: \url{https://doi.org/10.1145/2379776.2379785}
\BIBentrySTDinterwordspacing

\bibitem{2018-tsn-ehwsn}
\BIBentryALTinterwordspacing
K.~S. Adu{-}Manu, N.~H. Adam, C.~Tapparello, H.~Ayatollahi, and W.~B.
  Heinzelman, ``Energy-harvesting wireless sensor networks (eh-wsns): {A}
  review,'' \emph{{ACM} Trans. Sens. Networks}, vol.~14, no.~2, pp.
  10:1--10:50, 2018. [Online]. Available: \url{https://doi.org/10.1145/3183338}
\BIBentrySTDinterwordspacing

\bibitem{2021-tgcn-renew}
\BIBentryALTinterwordspacing
R.~V. Prasad, V.~S. Rao, C.~Sarkar, and I.~G. Niemegeers, ``Renew: {A}
  practical module for reliable routing in networks of energy-harvesting
  wireless sensors,'' \emph{{IEEE} Trans. Green Commun. Netw.}, vol.~5, no.~3,
  pp. 1558--1569, 2021. [Online]. Available:
  \url{https://doi.org/10.1109/TGCN.2021.3094771}
\BIBentrySTDinterwordspacing

\bibitem{2014-mobihoc-eh}
\BIBentryALTinterwordspacing
J.~Marasevic, C.~Stein, and G.~Zussman, ``Max-min fair rate allocation and
  routing in energy harvesting networks: Algorithmic analysis,'' vol.~78,
  no.~2, 2017, pp. 521--557. [Online]. Available:
  \url{https://doi.org/10.1007/s00453-016-0171-6}
\BIBentrySTDinterwordspacing

\bibitem{2018-infocom-ehwsn}
\BIBentryALTinterwordspacing
Q.~Chen, H.~Gao, Z.~Cai, L.~Cheng, and J.~Li, ``Energy-collision aware data
  aggregation scheduling for energy harvesting sensor networks,'' in \emph{2018
  {IEEE} Conference on Computer Communications, {INFOCOM} 2018, Honolulu, HI,
  USA, April 16-19, 2018}.\hskip 1em plus 0.5em minus 0.4em\relax {IEEE}, 2018,
  pp. 117--125. [Online]. Available:
  \url{https://doi.org/10.1109/INFOCOM.2018.8486366}
\BIBentrySTDinterwordspacing

\bibitem{2019-asplos-genesis}
\BIBentryALTinterwordspacing
G.~Gobieski, B.~Lucia, and N.~Beckmann, ``Intelligence beyond the edge:
  Inference on intermittent embedded systems,'' in \emph{Proceedings of the
  Twenty-Fourth International Conference on Architectural Support for
  Programming Languages and Operating Systems, {ASPLOS} 2019, Providence, RI,
  USA, April 13-17, 2019}, I.~Bahar, M.~Herlihy, E.~Witchel, and A.~R. Lebeck,
  Eds.\hskip 1em plus 0.5em minus 0.4em\relax {ACM}, 2019, pp. 199--213.
  [Online]. Available: \url{https://doi.org/10.1145/3297858.3304011}
\BIBentrySTDinterwordspacing

\bibitem{2022-tecs-energyconsumption}
\BIBentryALTinterwordspacing
S.~Ahmed, M.~Nawaz, A.~Bakar, N.~A. Bhatti, M.~H. Alizai, J.~H. Siddiqui, and
  L.~Mottola, ``Demystifying energy consumption dynamics in transiently powered
  computers,'' \emph{{ACM} Trans. Embed. Comput. Syst.}, vol.~19, no.~6, pp.
  47:1--47:25, 2020. [Online]. Available: \url{https://doi.org/10.1145/3391893}
\BIBentrySTDinterwordspacing

\bibitem{2019-pldi-samoyed}
\BIBentryALTinterwordspacing
K.~Maeng and B.~Lucia, ``Supporting peripherals in intermittent systems with
  just-in-time checkpoints,'' in \emph{Proceedings of the 40th {ACM} {SIGPLAN}
  Conference on Programming Language Design and Implementation, {PLDI} 2019,
  Phoenix, AZ, USA, June 22-26, 2019}, K.~S. McKinley and K.~Fisher, Eds.\hskip
  1em plus 0.5em minus 0.4em\relax {ACM}, 2019, pp. 1101--1116. [Online].
  Available: \url{https://doi.org/10.1145/3314221.3314613}
\BIBentrySTDinterwordspacing

\bibitem{2020-asplos-chrt}
\BIBentryALTinterwordspacing
J.~de~Winkel, C.~D. Donne, K.~S. Yildirim, P.~Pawelczak, and J.~D. Hester,
  ``Reliable timekeeping for intermittent computing,'' in \emph{{ASPLOS} '20:
  Architectural Support for Programming Languages and Operating Systems,
  Lausanne, Switzerland, March 16-20, 2020}, J.~R. Larus, L.~Ceze, and
  K.~Strauss, Eds.\hskip 1em plus 0.5em minus 0.4em\relax {ACM}, 2020, pp.
  53--67. [Online]. Available: \url{https://doi.org/10.1145/3373376.3378464}
\BIBentrySTDinterwordspacing

\bibitem{sensys-2017-task}
\BIBentryALTinterwordspacing
J.~D. Hester, K.~M. Storer, and J.~Sorber, ``Timely execution on intermittently
  powered batteryless sensors,'' in \emph{Proceedings of the 15th {ACM}
  Conference on Embedded Network Sensor Systems, SenSys 2017, Delft,
  Netherlands, November 06-08, 2017}, M.~R. Eskicioglu, Ed.\hskip 1em plus
  0.5em minus 0.4em\relax {ACM}, 2017, pp. 17:1--17:13. [Online]. Available:
  \url{https://doi.org/10.1145/3131672.3131673}
\BIBentrySTDinterwordspacing

\bibitem{2021-specification-BLE5}
\BIBentryALTinterwordspacing
M.~Woolley, ``Bluetooth® core specification version 5.0 feature
  enhancements,'' 2021. [Online]. Available:
  \url{https://www.bluetooth.com/bluetooth-resources/bluetooth-5-go-faster-go-further}
\BIBentrySTDinterwordspacing

\bibitem{1997-art-lcg}
D.~E. Knuth, \emph{The Art of Computer Programming, Volume 2: Seminumerical
  Algorithms}, 3rd~ed.\hskip 1em plus 0.5em minus 0.4em\relax Addison-Wesley,
  1997.

\bibitem{1962-siam-generator}
\BIBentryALTinterwordspacing
J.~L. Allard, A.~R. Dobell, and T.~E. Hull, ``Mixed congruential random number
  generators for decimal machines,'' \emph{J. {ACM}}, vol.~10, no.~2, pp.
  131--141, 1963. [Online]. Available:
  \url{https://doi.org/10.1145/321160.321163}
\BIBentrySTDinterwordspacing

\bibitem{2018-tecs-predic}
\BIBentryALTinterwordspacing
M.~M.~I. Rajib and A.~Nasipuri, ``Predictive retransmissions for intermittently
  connected sensor networks with transmission diversity,'' \emph{{ACM} Trans.
  Embed. Comput. Syst.}, vol.~17, no.~1, pp. 12:1--12:25, 2018. [Online].
  Available: \url{https://doi.org/10.1145/3092947}
\BIBentrySTDinterwordspacing

\bibitem{2021-mass-comm}
\BIBentryALTinterwordspacing
V.~Deep, M.~L. Wymore, A.~A. Aurandt, V.~Narayanan, S.~Fu, H.~Duwe, and
  D.~Qiao, ``Experimental study of lifecycle management protocols for
  batteryless intermittent communication,'' in \emph{{IEEE} 18th International
  Conference on Mobile Ad Hoc and Smart Systems, {MASS} 2021, Denver, CO, USA,
  October 4-7, 2021}.\hskip 1em plus 0.5em minus 0.4em\relax {IEEE}, 2021, pp.
  355--363. [Online]. Available:
  \url{https://doi.org/10.1109/MASS52906.2021.00052}
\BIBentrySTDinterwordspacing

\bibitem{2019-nca-tsch-survey}
\BIBentryALTinterwordspacing
S.~Kharb and A.~Singhrova, ``A survey on network formation and scheduling
  algorithms for time slotted channel hopping in industrial networks,''
  \emph{J. Netw. Comput. Appl.}, vol. 126, pp. 59--87, 2019. [Online].
  Available: \url{https://doi.org/10.1016/j.jnca.2018.11.004}
\BIBentrySTDinterwordspacing

\bibitem{2019-sensors-llns}
\BIBentryALTinterwordspacing
J.~V.~V. Sobral, J.~J. P.~C. Rodrigues, R.~A.~L. Rab{\^{e}}lo, J.~Al{-}Muhtadi,
  and V.~Korotaev, ``Routing protocols for low power and lossy networks in
  internet of things applications,'' \emph{Sensors}, vol.~19, no.~9, p. 2144,
  2019. [Online]. Available: \url{https://doi.org/10.3390/s19092144}
\BIBentrySTDinterwordspacing

\bibitem{2021-ieee-tsch}
\BIBentryALTinterwordspacing
N.~Choudhury, M.~M. Nasralla, P.~Gupta, and I.~U. Rehman, ``Centralized graph
  based {TSCH} scheduling for iot network applications,'' in \emph{2021 {IEEE}
  Intl Conf on Parallel {\&} Distributed Processing with Applications, Big Data
  {\&} Cloud Computing, Sustainable Computing {\&} Communications, Social
  Computing {\&} Networking (ISPA/BDCloud/SocialCom/SustainCom), New York City,
  NY, USA, September 30 - Oct. 3, 2021}.\hskip 1em plus 0.5em minus 0.4em\relax
  {IEEE}, 2021, pp. 1639--1644. [Online]. Available:
  \url{https://doi.org/10.1109/ISPA-BDCloud-SocialCom-SustainCom52081.2021.00219}
\BIBentrySTDinterwordspacing

\bibitem{2011-acm-ch}
\BIBentryALTinterwordspacing
L.~Tang, Y.~Sun, O.~Gurewitz, and D.~B. Johnson, ``{EM-MAC:} a dynamic
  multichannel energy-efficient {MAC} protocol for wireless sensor networks,''
  in \emph{Proceedings of the 12th {ACM} Interational Symposium on Mobile Ad
  Hoc Networking and Computing, MobiHoc 2011, Paris, France, May 16-20,
  2011}.\hskip 1em plus 0.5em minus 0.4em\relax {ACM}, 2011, p.~23. [Online].
  Available: \url{https://doi.org/10.1145/2107502.2107533}
\BIBentrySTDinterwordspacing

\bibitem{2014-ICS-iscc}
\BIBentryALTinterwordspacing
D.~D. Guglielmo, A.~Seghetti, G.~Anastasi, and M.~Conti, ``A performance
  analysis of the network formation process in {IEEE} 802.15.4e {TSCH} wireless
  sensor/actuator networks,'' in \emph{{IEEE} Symposium on Computers and
  Communications, {ISCC} 2014, Funchal, Madeira, Portugal, June 23-26,
  2014}.\hskip 1em plus 0.5em minus 0.4em\relax {IEEE} Computer Society, 2014,
  pp. 1--6. [Online]. Available:
  \url{https://doi.org/10.1109/ISCC.2014.6912607}
\BIBentrySTDinterwordspacing

\bibitem{2020-infocom-ost}
\BIBentryALTinterwordspacing
S.~Jeong, H.~Kim, J.~Paek, and S.~Bahk, ``{OST:} on-demand {TSCH} scheduling
  with traffic-awareness,'' in \emph{39th {IEEE} Conference on Computer
  Communications, {INFOCOM} 2020, Toronto, ON, Canada, July 6-9, 2020}.\hskip
  1em plus 0.5em minus 0.4em\relax {IEEE}, 2020, pp. 69--78. [Online].
  Available: \url{https://doi.org/10.1109/INFOCOM41043.2020.9155496}
\BIBentrySTDinterwordspacing

\bibitem{2020-algo-wsnsurvey}
\BIBentryALTinterwordspacing
C.~T. Nakas, D.~Kandris, and G.~Visvardis, ``Energy efficient routing in
  wireless sensor networks: {A} comprehensive survey,'' \emph{Algorithms},
  vol.~13, no.~3, p.~72, 2020. [Online]. Available:
  \url{https://doi.org/10.3390/a13030072}
\BIBentrySTDinterwordspacing

\bibitem{2021-csur-sync}
\BIBentryALTinterwordspacing
M.~Zimmerling, L.~Mottola, and S.~Santini, ``Synchronous transmissions in
  low-power wireless: {A} survey of communication protocols and network
  services,'' \emph{{ACM} Comput. Surv.}, vol.~53, no.~6, pp. 121:1--121:39,
  2021. [Online]. Available: \url{https://doi.org/10.1145/3410159}
\BIBentrySTDinterwordspacing

\bibitem{2009-icnp-esc}
\BIBentryALTinterwordspacing
Y.~Gu, T.~Zhu, and T.~He, ``{ESC:} energy synchronized communication in
  sustainable sensor networks,'' in \emph{Proceedings of the 17th annual {IEEE}
  International Conference on Network Protocols, 2009. {ICNP} 2009, Princeton,
  NJ, USA, 13-16 October 2009}, H.~Schulzrinne, K.~K. Ramakrishnan, T.~G.
  Griffin, and S.~V. Krishnamurthy, Eds.\hskip 1em plus 0.5em minus 0.4em\relax
  {IEEE} Computer Society, 2009, pp. 52--62. [Online]. Available:
  \url{https://doi.org/10.1109/ICNP.2009.5339699}
\BIBentrySTDinterwordspacing

\bibitem{2011-mobihoc-esa}
\BIBentryALTinterwordspacing
L.~Huang and M.~J. Neely, ``Utility optimal scheduling in energy-harvesting
  networks,'' \emph{{IEEE/ACM} Trans. Netw.}, vol.~21, no.~4, pp. 1117--1130,
  2013. [Online]. Available: \url{https://doi.org/10.1109/TNET.2012.2230336}
\BIBentrySTDinterwordspacing

\bibitem{2022-tmc-das}
\BIBentryALTinterwordspacing
T.~Zhu, J.~Li, H.~Gao, and Y.~Li, ``Data aggregation scheduling in battery-free
  wireless sensor networks,'' \emph{{IEEE} Trans. Mob. Comput.}, vol.~21,
  no.~6, pp. 1972--1984, 2022. [Online]. Available:
  \url{https://doi.org/10.1109/TMC.2020.3035671}
\BIBentrySTDinterwordspacing

\bibitem{2016-infocom-ehdataprocess}
\BIBentryALTinterwordspacing
S.~Yang, Y.~Tahir, P.~Chen, A.~Marshall, and J.~A. McCann, ``Distributed
  optimization in energy harvesting sensor networks with dynamic in-network
  data processing,'' in \emph{35th Annual {IEEE} International Conference on
  Computer Communications, {INFOCOM} 2016, San Francisco, CA, USA, April 10-14,
  2016}.\hskip 1em plus 0.5em minus 0.4em\relax {IEEE}, 2016, pp. 1--9.
  [Online]. Available: \url{https://doi.org/10.1109/INFOCOM.2016.7524475}
\BIBentrySTDinterwordspacing

\bibitem{2023-mtp-smartFa}
A.~D. Dhruva, B.~Prasad, S.~Kamepalli, S.~Kunisetti \emph{et~al.}, ``An
  efficient mechanism using iot and wireless communication for smart farming,''
  \emph{Materials Today: Proceedings}, vol.~80, pp. 3691--3696, 2023.

\bibitem{2008-cc-rfid}
\BIBentryALTinterwordspacing
B.~Sun, Y.~Xiao, C.~Li, H.~Chen, and T.~A. Yang, ``Security co-existence of
  wireless sensor networks and {RFID} for pervasive computing,'' \emph{Comput.
  Commun.}, vol.~31, no.~18, pp. 4294--4303, 2008. [Online]. Available:
  \url{https://doi.org/10.1016/j.comcom.2008.05.035}
\BIBentrySTDinterwordspacing

\bibitem{2020-sensors-smartAg}
\BIBentryALTinterwordspacing
K.~Haseeb, I.~U. Din, A.~Almogren, and N.~Islam, ``An energy efficient and
  secure iot-based {WSN} framework: An application to smart agriculture,''
  \emph{Sensors}, vol.~20, no.~7, p. 2081, 2020. [Online]. Available:
  \url{https://doi.org/10.3390/s20072081}
\BIBentrySTDinterwordspacing

\bibitem{2021-cc-edfg}
\BIBentryALTinterwordspacing
D.~K. Sah, K.~Cengiz, P.~K. Donta, V.~N. Inukollu, and T.~Amgoth, ``{EDGF:}
  empirical dataset generation framework for wireless sensor networks,''
  \emph{Comput. Commun.}, vol. 180, pp. 48--56, 2021. [Online]. Available:
  \url{https://doi.org/10.1016/j.comcom.2021.08.017}
\BIBentrySTDinterwordspacing

\bibitem{2018-tosn-ehwsn}
\BIBentryALTinterwordspacing
K.~S. Adu{-}Manu, N.~H. Adam, C.~Tapparello, H.~Ayatollahi, and W.~B.
  Heinzelman, ``Energy-harvesting wireless sensor networks (eh-wsns): {A}
  review,'' \emph{{ACM} Trans. Sens. Networks}, vol.~14, no.~2, pp.
  10:1--10:50, 2018. [Online]. Available: \url{https://doi.org/10.1145/3183338}
\BIBentrySTDinterwordspacing

\bibitem{2015-commag-ehwsn}
\BIBentryALTinterwordspacing
D.~Wu, J.~He, H.~Wang, C.~Wang, and R.~Wang, ``A hierarchical packet forwarding
  mechanism for energy harvesting wireless sensor networks,'' \emph{{IEEE}
  Commun. Mag.}, vol.~53, no.~8, pp. 92--98, 2015. [Online]. Available:
  \url{https://doi.org/10.1109/MCOM.2015.7180514}
\BIBentrySTDinterwordspacing

\bibitem{2019-ipsn-ehs}
\BIBentryALTinterwordspacing
K.~Geissdoerfer, R.~Jurdak, B.~Kusy, and M.~Zimmerling, ``Getting more out of
  energy-harvesting systems: energy management under time-varying utility with
  preact,'' in \emph{Proceedings of the 18th International Conference on
  Information Processing in Sensor Networks, {IPSN} 2019, Montreal, QC, Canada,
  April 16-18, 2019}, M.~R. Eskicioglu, L.~Mottola, and B.~Priyantha,
  Eds.\hskip 1em plus 0.5em minus 0.4em\relax {ACM}, 2019, pp. 109--120.
  [Online]. Available: \url{https://doi.org/10.1145/3302506.3310393}
\BIBentrySTDinterwordspacing

\end{thebibliography}

\begin{IEEEbiography}[{\includegraphics[width=1in,height=1.25in,clip,keepaspectratio]{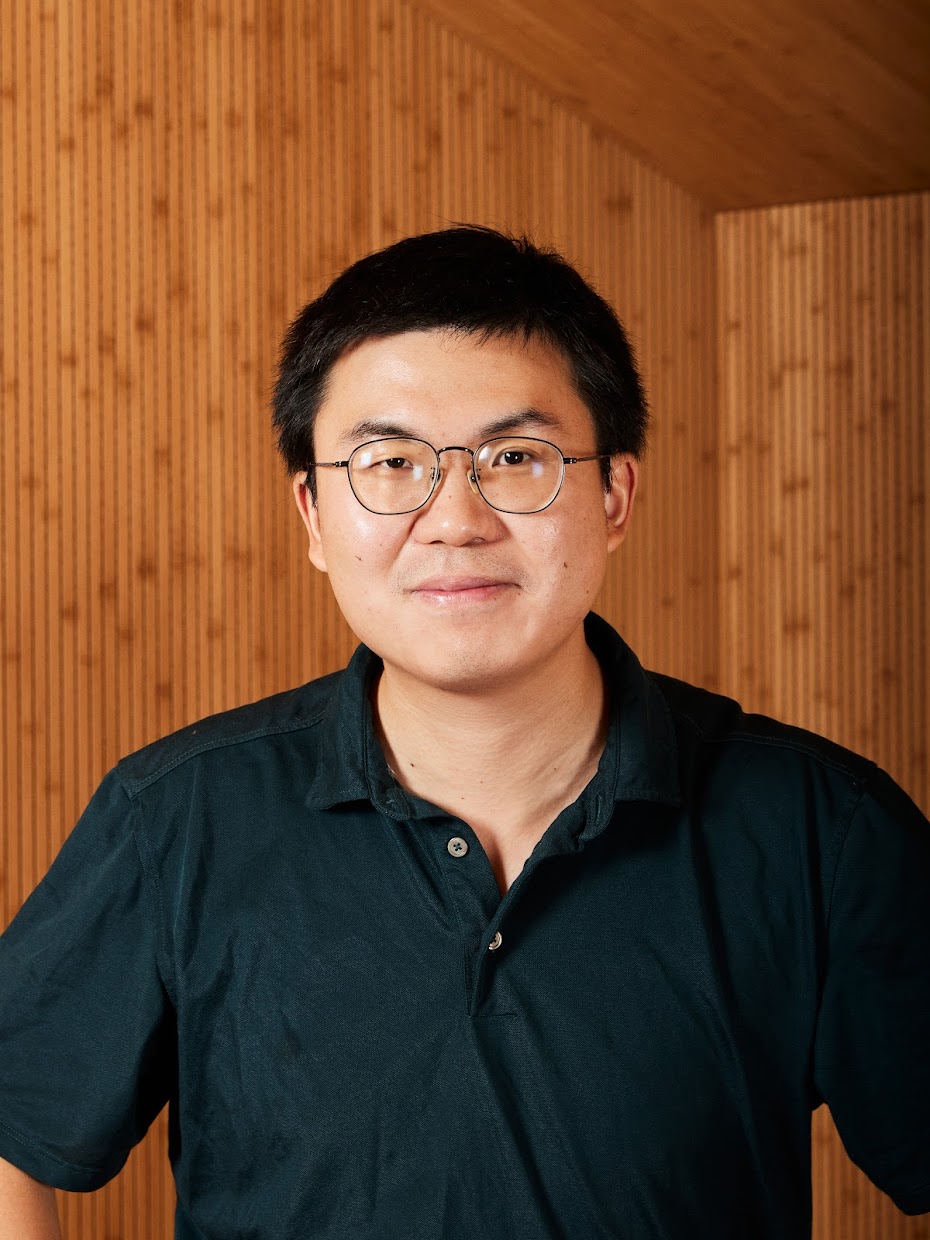}}]{Gaosheng Liu}
is currently a PhD candidate at Vrije Universiteit Amsterdam. He received his M.S. degree in the College of intelligence and computing, Tianjin University, Tianjin, China, in 2019. He received his bachelor degree from Tianjin university of technology and education, Tianjin, China, in 2015. His research interests include
intermittent communication, cyber-physical systems, and AIoT.
\end{IEEEbiography}

\begin{IEEEbiography}[{\includegraphics[width=1in,height=1.25in,clip,keepaspectratio]{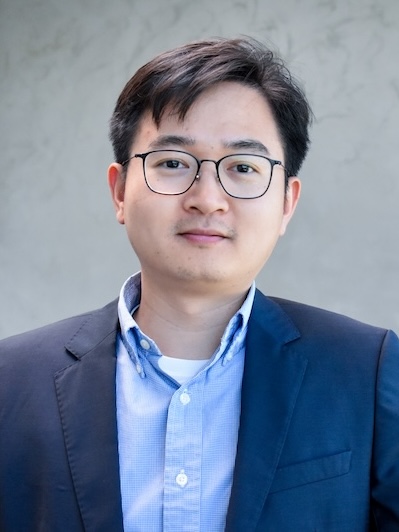}}]{Lin Wang}
is currently a Full Professor and Head of the Computer Networks group in the Department of Computer Science at Paderborn University, Germany. He received his Ph.D. in Computer Science from the Institute of Computing Technology, Chinese Academy of Sciences in 2015. Before joining Paderborn University, he was a tenured Assistant Professor at Vrije Universiteit Amsterdam and held positions at TU Darmstadt, SnT Luxembourg, and IMDEA Networks Institute. His research is focused on networked systems, with the goal of achieving efficiency, scalability, and sustainability. He has served as a referee for several funding agencies (e.g., DFG, ISF, and HK-RGC), on the program committees of major conferences (e.g., SoCC, Middleware, ICCD, INFOCOM, MobiHoc, ICDCS, IWQoS, and SEC) and as a reviewer for top journals (e.g., ToN, JSAC, TMC, and TPDS). He has received several important awards, including a Google Research Scholar Award, an Outstanding Paper Award of IEEE RTSS 2022, Best Paper Awards of IEEE IPCCC 2023, IEEE ISCC 2024, and IEEE HotPNS 2016, and an Athene Young Investigator Award of TU Darmstadt. He is currently a Senior Member of IEEE. 
\end{IEEEbiography}

\vfill

\end{document}